\documentclass[aos,preprint]{imsart}

\RequirePackage[numbers]{natbib}
\usepackage{amsthm,amsmath,amssymb,comment,tikz,pgf,bm}
\usepackage{amsopn, amsfonts, leftidx}
\usetikzlibrary{shapes}

\usepackage{algorithm}
\usepackage{algpseudocode}
\usepackage{graphicx}

\RequirePackage[colorlinks,citecolor=blue,urlcolor=blue]{hyperref}

\startlocaldefs
\endlocaldefs

\newtheorem{dfn}{Definition}
\newtheorem{thm}{Theorem}[section]
\newtheorem{lem}{Lemma}[section]

\newtheorem{cor}{Corollary}[section]

    \newtheoremstyle{TheoremNum}
        {\topsep}{\topsep}              
        {\itshape}                      
        {}                              
        {\bfseries}                     
        {.}                             
        { }                             
        {\thmname{#1}\thmnote{ \bfseries #3}}
    \theoremstyle{TheoremNum}

\newtheorem{thmn}{Theorem}
\newtheorem{lemn}{Lemma}

\DeclareMathOperator{\rel}{rel}

\DeclareMathOperator{\sop}{so}
\DeclareMathOperator{\sink}{sink}
\DeclareMathOperator{\ine}{in}
\DeclareMathOperator{\oute}{out}

\DeclareMathOperator{\doo}{do}
\DeclareMathOperator{\dis}{dis}

\DeclareMathOperator{\pa}{pa}
\DeclareMathOperator{\de}{de}
\DeclareMathOperator{\si}{si}
\DeclareMathOperator{\ch}{ch}
\DeclareMathOperator{\an}{an}

\DeclareMathOperator{\cb}{cb}

\DeclareMathOperator{\pre}{pre}

\def\ci{\perp\!\!\!\perp}

\begin{document}

\begin{frontmatter}

\title{Causal Inference with a Graphical Hierarchy of Interventions}
\runtitle{Hierarchy of Interventions}


\begin{aug}
\author{\fnms{Ilya} \snm{Shpitser}\thanksref{m1}
\ead[label=e1]{ilyas@cs.jhu.edu}
}
\and
\author{\fnms{Eric} \snm{Tchetgen Tchetgen}\thanksref{m2}
\ead[label=e2]{etchetge@hsph.harvard.edu}
}
\runauthor{I. Shpitser and E. Tchetgen Tchetgen}

\affiliation{
Johns Hopkins University\thanksmark{m1}
and Harvard University\thanksmark{m2}
}

\address{
Department of Computer Science\\
Johns Hopkins University\\
3400 N Charles Street\\
Baltimore, Maryland 21218\\
\printead{e1}
}

\address{
School of Public Health\\
Harvard University\\
677 Huntington Avenue\\
Kresge Building\\
Boston, Massachusetts 02115\\
\printead{e2}
}
\end{aug}

\begin{abstract}
Identifying causal parameters
from observational data is fraught with subtleties due to the
issues of selection bias and confounding.  In addition, more complex questions
of interest, such as effects of treatment on the treated and mediated
effects may not always be identified even in data where treatment assignment
is known and under investigator control, or may be identified under one causal
model but not another.

Increasingly complex effects of interest, coupled with a diversity of
causal models in use resulted in a fragmented view of identification.
This fragmentation makes it unnecessarily difficult to determine if a given
parameter is identified (and in what model), and what assumptions must hold for
this to be the case.  This, in turn, complicates the development of estimation
theory and sensitivity analysis procedures.

In this paper, we give a unifying view of a large class of
causal effects of interest, including novel effects not previously considered,
in terms of a hierarchy of interventions, and show that identification theory
for this large class reduces to an identification theory of random variables
under interventions from this hierarchy.  Moreover, we show that one type of
intervention in the hierarchy is naturally associated with queries identified
under the Finest Fully Randomized Causally Interpretable Structure Tree Graph
(FFRCISTG) model of Robins (via the extended g-formula),
and another is naturally associated with queries
identified under the Non-Parametric Structural Equation Model with Independent
Errors (NPSEM-IE) of Pearl, via a more general functional
we call the edge g-formula.

Our results motivate the study of estimation theory for the edge g-formula,
since we show it arises both in mediation analysis, and in settings where
treatment assignment has unobserved causes, such as models associated with
Pearl's front-door criterion.

\end{abstract}



\end{frontmatter}


\section{Introduction}

The goal of the empirical sciences is discerning cause-effect relationships
by experimentation and analysis.  This is made difficult by the ubiquity of
hidden variables, and the difficulty of collecting data free from confounding
and selection bias.  Two useful frameworks for addressing these
difficulties have been potential outcomes, introduced by Neyman
\cite{neyman23app}, and expanded by Rubin \cite{rubin74potential}, and causal
graphical models, first used in linear models by Wright
\cite{wright21correlation}, and later expanded into a general framework (see for
example \cite{spirtes01causation}, and \cite{pearl09causality}).
There exists a modern synthesis of these two frameworks, where causal
models based on non-parametric structural equations are defined on potential
outcome random variables, and assumptions defining these models can be
represented by (absences) of arrows in a graph.  See
\cite{pearl09causality} chapter 7, and \cite{thomas13swig} for a detailed
treatment.

Potential outcome random variables represent outcomes under
a hypothetical \emph{intervention} operation, which corresponds to
an idealized randomized control trial.  Concepts such as the overall
causal effect of a treatment can be represented as causal parameters
on appropriate potential outcomes, and as statistical estimands if appropriate
assumptions hold.

The synthesis of potential outcomes and graphs has been instrumental in much
of the recent work on identification of various types of causal parameters
such as total effects \cite{robins86new, tian02on, shpitser06id,shpitser06idc,
shpitser07hierarchy}, and mediated effects \cite{pearl01direct,
chen05ijcai, shpitser13cogsci}.

Nevertheless, the existing literature suffers from three problems.  First,
a single graph may correspond to different causal models, which means a
particular causal parameter may be identified under one causal model, but not
under another, even though the models \emph{share the same graph}.
Second, different types of causal parameters seem to have different key issues
underlying their identification, which makes it difficult to determine the
specific assumptions that must hold for identification.  For
instance, certain types of unobserved confounding must be absent in order for
overall effects to be identifiable, while even completely unconfounded
mediated effects may be unidentified \cite{chen05ijcai}.  
Finally, because of the complex nature of identification theory for causal
parameters, existing conventional wisdom on what is identifiable is
\emph{too conservative}.  For example, it is often assumed that
a mediator and outcome must remain completely unconfounded in order to obtain
identification of mediated causal effects.  However, this is not true
\cite{shpitser13cogsci}.

These issues make it difficult to determine \emph{if} a particular causal
parameter is identified, and \emph{under what model}, what \emph{assumptions}
underlie this identification, and what the corresponding statistical parameter
is.  This complicates estimation theory, the development of parametric
relaxations that permit identification, and sensitivity analysis procedures.

\subsection{Outline of the Paper}

The contents of the paper can be summarized by a picture in
Fig. \ref{fig:hierarchy}.  In section \ref{sec:intro}, we introduce our
notation, necessary graph theory, standard interventions
(which we call node interventions in this manuscript) and potential outcomes,
which are responses to node interventions.  We also introduce
the FFRCISTG model of Robins, which in this paper we call the ``single world
model (SWM),''
and the NPSEM-IE of Pearl, which is a submodel of the FFRCISTG model, and
which we call the ``multiple worlds model (MWM).''  The reasons for these
names will become clear when these models are defined.
The subset relationship of these two models is shown explicitly in
Fig. \ref{fig:hierarchy}.  Finally, we discuss targets of interest in causal
inference known as total effects, which are defined in terms of node
interventions, and discuss identification theory for these targets under
the SWM via the extended g-formula.

In section \ref{sec:edge-path}, we define additional types of interventions,
that we term edge and path interventions, and responses to these types of
interventions via recursive substitution.  Responses to node, edge and path
interventions form an inclusion hierarchy in the sense that responses to
node interventions are a special case of responses to edge interventions, which
are in turn a special case of responses to path interventions.  This inclusion
is denoted by the subset relations in Fig. \ref{fig:hierarchy}.  We also
discuss how targets of inference in mediation analysis known as direct and
indirect effects are defined in terms of edge interventions.

In section \ref{sec:targets}, we show how we can express a wide variety of
targets of interest in causal inference, such as path-specific effects (PSEs)
or effects of treatment on the multiply treated (ETMTs)
as responses to path interventions.  In addition, we show that path
interventions are general enough to accommodate novel targets which combine
features of PSEs and ETMTs, which we call effects of treatment on the indirectly
treated (ETITs).  Our results then imply novel identification results
for these targets, and others not previously considered in the literature, but
expressible as path interventions.

In section \ref{sec:id}, we show that there is a natural correspondence
between causal models and intervention types we discuss in the following sense.
We show that responses to node interventions are identified under the SWM,
and responses to edge interventions are identified under the
MWM.  Furthermore, we show that if a response to an edge intervention
cannot be expressed as a node intervention, then it is \emph{not} identified
under the SWM, and if a response to a path intervention cannot be
expressed as an edge intervention, then it is \emph{not} identified under the
MWM.

The identification of node interventions under the SWM is via
the well known \emph{extended g-formula} \cite{robins04effects,thomas13swig},
which we give as equation (\ref{eqn:g-formula}).
The identification of edge interventions under the MWM is via
a generalization of (\ref{eqn:g-formula}), which we call the
\emph{edge g-formula}, and give as equation (\ref{eqn:g-formula-edge}).

We also give examples of targets of interest in causal inference that do not
correspond to responses to path interventions, as well as an example of
a submodel of the MWM where even path interventions not ordinarily identified
under the MWM are identified.

In Section \ref{sec:swigs} we briefly discuss the relationship of our
results to Single World Intervention Graphs (SWIGs) \cite{thomas13swig}.

Section \ref{sec:tian}
shows that a certain class of functionals that identify causal
effects in latent variable causal models \cite{tian02on,shpitser06id}
corresponds to functionals derived from the edge g-formula.
This implies, in particular, that functionals that arise for treatment effects
with unobserved causes of treatments, such as the front-door functional,
also arise in mediation analysis.

In section \ref{sec:estimator},
we illustrate the connection of our work to existing estimation theory for
causal parameters, and suggest avenues of future work,
by giving a known example of an estimator for a parameter derived from a special
case of the edge g-formula.




What the overall picture implies is that once we solve the identification
problem for the responses to interventions in our hierarchy, as we do
here, we immediately reduce the identification problem for a wide class of
targets of interest to the much easier problem of translating those targets
into responses to path interventions.
Once that translation is complete, the question of what is identified under
what model is immediately settled.  In addition, our developments imply that
estimation theory for functionals derived from the edge g-formula
is relevant for a large class of inference targets identified under the
MWM, including path-specific effects, effects of treatment on the multiply
treated, and certain total causal effects with unobserved causes of treatments.

In the interests of space, the vast majority of arguments for our results appear
in the appendices in the supplementary materials \cite{shpitser15nep}.
In addition, the supplementary materials contains our rationale for the use of
path interventions, rather than simpler or more algebraic representations of
causal inference targets.

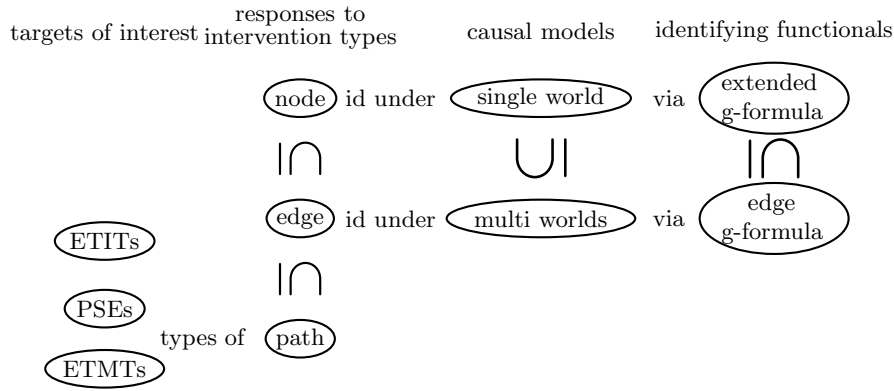
\begin{figure}
\begin{center}
\begin{tikzpicture}[>=stealth, node distance=1.6cm]
    \tikzstyle{format} = [draw, thick, ellipse, minimum size=5.0mm,
	inner sep=0pt]
    \tikzstyle{format2} = [draw, thick, dashed, ellipse, minimum size=5.0mm,
	inner sep=0pt]

	\begin{scope}
		\path[->]
			node[] (dummy) {}
			node[format, above of=dummy,yshift=-0.3cm] (new) {ETITs}
			node[format, above of=dummy,yshift=-1.2cm] (med) {PSEs}
			node[format, below of=dummy,yshift=+1.2cm] (ett) {ETMTs}
			node[format, right of=dummy, xshift=1.0cm] (path) {path}
			node[format, above of=path] (edge) {edge}
			node[format, above of=edge] (node) {node}
			node[above of=node,yshift=-0.5cm] (l1)
				{responses to}
			node[above of=node,yshift=-0.8cm] (l2)
				{intervention types}
			node[format, right of=edge, xshift=1.6cm]
				(judea) {multi worlds}
			node[format, right of=node, xshift=1.6cm]
				(jamie) {single world}
			node[above of=jamie,yshift=-0.7cm] (l3)
				{causal models}
			node[above of=med,yshift=2.05cm] (l4)
				{targets of interest}

			node[format, right of=jamie, xshift=1.5cm,
				text width=1.4cm]
				(g-formula) {extended g-formula}

			node[above of=g-formula,yshift=-0.7cm] (l3)
				{identifying functionals}

			node[format, right of=judea, xshift=1.5cm,
				text width=1.4cm]
				(e-g-formula) {\hspace{3.5mm}edge g-formula}

			node[right of=jamie,xshift=0.1cm] {via}
			node[right of=judea,xshift=0.1cm] {via}

			node[right of=dummy,xshift=-0.3cm] {types of}

			node[right of=node,xshift=-0.4cm] {id under}
			node[right of=edge,xshift=-0.4cm] {id under}

			node[below of=g-formula, yshift=0.8cm]
				{\rotatebox{270}{\Huge $\subseteq$}}
			node[below of=jamie, yshift=0.8cm]
				{\rotatebox{90}{\Huge $\subseteq$}}

			node[below of=node, yshift=0.8cm]
				{\rotatebox{270}{\huge $\subseteq$}}
			node[below of=edge, yshift=0.8cm]
				{\rotatebox{270}{\huge $\subseteq$}}


		;
	\end{scope}
\end{tikzpicture}
\end{center}
\caption{A hierarchy of responses to interventions defined with respect to
features of a causal graph, the relationship of this hierarchy to targets of
interest in causal inference, such as path-specific effects (PSEs),
effects of treatments on the multiply treated (ETMTs), and new targets such
as effects of treatments on the indirectly treated (ETITs),
and identifiability under causal models defined in the literature.
}
\label{fig:hierarchy}
\end{figure}

\section{Notation and Definitions}
\label{sec:intro}
We introduce graph theory terms, potential outcomes, and statistical and
causal graphical models.

\subsection{Graphs and Random Variables}
\label{sec:basic-dfn}

We will associate random variables with vertices in graphs.  We will denote
\emph{both} a single vertex and a single corresponding random variable as
an uppercase Roman letter, e.g. $A$.  Sets of vertices (and corresponding
random variables) will be denoted by uppercase bold letters, e.g. ${\bf A}$.

For a random variable $V$, let
${\mathfrak X}_V$ be the state space of $V$.  For example if $V$ is binary,
then ${\mathfrak X}_V = \{ 0, 1 \}$.
We denote elements of a set ${\mathfrak X}_A$ (values of $A$)
by lowercase Roman letters: $a \in {\mathfrak X}_A$.
The state space of a set ${\bf V}$ of random
variables is simply the Cartesian product of the individual state spaces:
${\mathfrak X}_{\bf V} =
\times_{V \in {\bf V}} \left( {\mathfrak X}_{V} \right)$.

Sets of values corresponding to sets of
random variables will be denoted by lowercase bold letters, e.g.
${\bf a} \in {\mathfrak X}_{\bf A}$.
Sometimes we will denote a restriction of
a set of values by a set subscript.  That is if ${\bf v}$ is a set of values
of ${\bf V}$, and ${\bf A} \subseteq {\bf V}$, then
${\bf v}_{\bf A}$ is a restriction of ${\bf v}$ to ${\bf A}$.

An edge in a graph is a vertex adjacency coupled with an orientation.
A path in a directed graph is a (possibly empty) sequence of nodes of
the form $(A_1A_2A_3\ldots A_{k-1}A_k)$, where each node in the sequence
occurs exactly once, and each $A_{i},A_{i+1}$ share an edge.
The first vertex in a path sequence is called the source, and the last vertex
is called the sink.  A path with two vertices $(A_1A_2)$ is just an edge.

A subpath of a path is a subsequence of edges in a path that themselves
form a path.  A suffix subpath of
$(A_1A_2\ldots A_{m-1}A_{m}\ldots A_{k-1}A_k)$
is a subpath of the form $(A_{m-1}A_{m}\ldots A_{k-1}A_k)$, while
a prefix subpath is a subpath of the form
$(A_1A_2\ldots A_{m-1}A_{m})$.
A directed path from $A_1$ to $A_k$ has edges for every $i$ of the form
$A_i \to A_{i+1}$.  We will denote a directed path as
$(A_1A_2\ldots A_k)_{\to}$, and also by Greek letters, e.g. $\alpha$,
and sets of directed paths by bold Greek letters, e.g. ${\bm\alpha}$.
A source vertex of $\alpha$ will be written $\sop_{\cal G}(\alpha)$,
and the sink vertex will be written $\sink_{\cal G}(\alpha)$.

We say a directed cycle exists in a graph if it contains a path
$(A_1A_2A_3 \ldots A_{k})_{\to}$ and an edge $(A_{k}A_1)_{\to}$.
A directed graph lacking directed cycles is called acyclic, abbreviated as
DAG.

\subsection{
Causal Models of a DAG}
\label{sec:models}

For a subset ${\bf A}$ of random variables ${\bf V}$,
and a value assignment ${\bf a}$ to
${\bf A}$, we denote a forced assignment of ${\bf A}$ to an element of
${\mathfrak X}_{\bf A}$ as a \emph{node intervention}.
A node intervention which maps ${\bf A}$ to ${\bf a} \in
{\mathfrak X}_{\bf A}$
will be denoted by $\nu_{\bf a}$.  Pearl denoted node interventions
$\nu_{\bf a}$ by $\doo({\bf a})$, and Robins by $g = {\bf a}$.  We use
alternative notation in this paper to avoid ambiguity, because we will consider
other types of interventions.  It is also possible to consider more complex
types of interventions on nodes, known as \emph{dynamic treatment regimes},
where assigned values to ${\bf A}$ are not constants, but functions of variables
assigned and observed in the past
\cite{robins86new,murphy03optimal,moodie07demystifying}.  Although
generalizations of our results to this setting are possible, we do not pursue
them in the interests of space.

For a random variable $Y \in {\bf V}$, and ${\bf a} \in {\mathfrak X}_{\bf A}$
for a set ${\bf A} \subseteq {\bf V}$,
we denote a (random) response to a node intervention $\nu_{\bf a}$ as
$Y({\bf a})$.  These random variables are also called
\emph{potential outcomes}, because $Y$ is often an outcome of interest, and
the intervention is often hypothetical, rather than actually occurring.
Given a set ${\bf Y} = \{ Y_1, \ldots, Y_k \}$ of random variables, we denote
$\{ Y_1({\bf a}), \ldots, Y_k({\bf a}) \}$ by ${\bf Y}({\bf a})$ or
$\{ {\bf Y} \}({\bf a})$.

Let $\pa_{\cal G}(V)$ be the set of parents of $V$ in ${\cal G}$, that is
the set $\{ W \mid (WV)_{\to} \text{ is in } {\cal G} \}$.
Following \cite{thomas13swig}, given a DAG ${\cal G}$ with vertices ${\bf V}$,
we will assume the existence of $V({\bf v}_{\pa_{\cal G}(V)})$
for every $V \in {\bf V}$, and for all
${\bf v}_{\pa_{\cal G}(V)} \in {\mathfrak X}_{\pa_{\cal G}(V)}$, as well as
a well-defined joint distribution over these random variables, and use
these potential outcomes, and the associated joint,
to define others using recursive substitution.

In particular, for any ${\bf A} \subseteq {\bf V}$, and any
${\bf a} \in {\mathfrak X}_{\bf A}$, we define for every $V \in {\bf V}$
\begin{align}
V({\bf a}) \equiv V({\bf a}_{\pa_{\cal G}(V)},
	\{ \pa_{\cal G}(V) \setminus {\bf A} \}({\bf a}))
\label{eqn:rec-sub}
\end{align}
In words, this states that the response of $V$ to $\nu_{\bf a}$ is defined
as the potential outcome where all parents of $V$ which are in ${\bf A}$
are assigned an appropriate value from ${\bf a}$, and all
other parents are assigned whatever value they would have attained under
a node intervention $\nu_{\bf a}$ (these are defined recursively, and the
definition terminates because of the lack of directed cycles in ${\cal G}$).
For example, in the graph in Fig. \ref{fig:triangle} (a),
$Y(a) = Y(a,M(a))$.

It is possible to construct additional types of potential outcomes
other than those that are responses to node interventions.  We will discuss
some such potential outcomes later.
However, responses to node interventions are sufficient to define
causal models.  Just as a statistical model is a set of distributions over
${\bf V}$ defined by some restriction, we view a causal model as a set of
distributions over $\{ V({\bf v}_{\pa_{\cal G}(V)}) \mid V \in {\bf V} \}$
defined by some restriction.  We will call elements of a causal model
\emph{causal structures}, and denote them as $c({\bf V},{\cal G})$, by analogy
with $p({\bf V})$, but indexed by a graph.
In this paper we will consider two causal models.

We adopt the definitions presented in \cite{thomas13swig}.
We define the \emph{finest fully randomized causally interpretable structured
tree graph (FFRCISTG) model}
associated with a DAG ${\cal G}$ with vertices ${\bf V}$,
as the set of all possible potential outcome responses subject to
the restriction that the variables in the set
\[
\left\{ V({\bf v}_{\pa_{\cal G}(V)}) \mid 
V \in {\bf V} \right\}
\]
are mutually independent for every ${\bf v} \in {\mathfrak X}_{\bf V}$.
We define the \emph{non-parametric structural equation model with independent
errors (NPSEM-IE)} associated with a DAG ${\cal G}$ with vertices ${\bf V}$,
as the set of all possible potential outcome responses subject to the
restriction that the sets of variables
\[
\left\{
\{ V({\bf a}_V) \mid {\bf a}_V \in {\mathfrak X}_{\pa_{\cal G}(V)}
\} \middle| V \in {\bf V}
\right\}
\]
are mutually independent.
The NPSEM-IE associated with a particular graph is a submodel of the FFRCISTG
model associated with the same graph, because it always places at
least as many restrictions on potential outcome responses, and in most cases
many more.

For example,
the binary FFRCISTG model associated with the DAG in Fig. \ref{fig:triangle} (a)
asserts that variables $W$, $A(w)$, $M(a,w)$, $Y(a,m)$ are mutually independent
for any $a,m,w \in \{ 0, 1 \}$,
while the binary NPSEM-IE model associated with the same DAG
asserts that sets $\{ W \}, \{ A(w) \mid w \in \{0,1\} \},
\{ M(a,w) \mid a \in \{0,1\},w \in \{0,1\} \}, \{ Y(a,m) \mid
	a \in \{0,1\}, m \in \{0,1\} \}$ are mutually
independent.
The FFRCISTG model always imposes restrictions on a set of variables under
a single set of interventions (a ``single world''), while the NPSEM-IE may also
impose restrictions on variables across multiple conflicting sets of
interventions simultaneously.
To emphasize this, we will refer to the FFRCISTG model as a
``single world model'' (SWM), and to the NPSEM-IE as a
``multiple worlds model'' (MWM) in the remainder of this paper.

A crucial difference between the SWM and the MWM, is that the assumptions of
the former are possible to test, at least in principle, by checking
independences in a distribution of responses in an idealized randomized
controlled trial.  That is, if we wanted to check if $W$ is independent of
$A(w)$, we could check independence in a joint distribution obtained from
recording, for a set of units,
the values of $W$ immediately before treatment $w$ is assigned, and
the response values of $A$ under that assignment.  However, checking if
$M(a)$ is independent of $Y(a',m)$ would entail somehow knowing how the
response $M$ of a unit behaves under assigned treatment $a$, and
\emph{simultaneously} how the response $Y$ of the unit behaves under
a conflicting treatment $a'$ (and $m$).
One may be able to argue for explicit construction of such
joint responses in certain designs \cite{imai13experimental}, or
for certain types of units, for instance logic gates in a digital circuit.
However, in general, assumptions defining the MWM are not experimentally
testable.

\subsection{Identification of Node Interventions}
\label{sec:node-id}

Responses to interventions of various types can be used to define
targets of interest, discussed in more detail in Section \ref{sec:targets}.
However, in order for these definitions to be useful, they must be linked to
actually observed data.  If such a link can be provided, that is, if
a particular response can be expressed as a functional
of the observed joint distribution $p({\bf V})$ for any element of
a causal model, we say that the response is \emph{identified} under that causal
model from $p({\bf V})$.

In causal models, this link is typically provided via
the \emph{consistency assumption}, which is sometimes informally stated as
``in the subpopulation where ${\bf A} = {\bf a}$, ${\bf Y}({\bf a})$
behaves as ${\bf Y}$.''
Under the definition of the SWM (and the MWM),
consistency is implied by (\ref{eqn:rec-sub}), see \cite{thomas13swig}, p. 21.
Thus, consistency is ``folded in'' to the model definition.  Thus we will
describe identification in terms of a particular model, and not mention
consistency itself.  Note that (\ref{eqn:rec-sub}) is an assumption defined
using a particular graph.  If we are mistaken about the true graph, for
instance due to the presence of unaccounted hidden variables, then some parts
of (\ref{eqn:rec-sub}), and thus some parts of the consistency assumption, may
not be justifiable under the true causal model.

Identification theory for node interventions in causal DAG models is well
understood.
Given a DAG ${\cal G}$ with vertices ${\bf V}$, and two arbitrary subsets
${\bf A},{\bf Y}$ of ${\bf V}$ (not necessarily disjoint), the distribution
$p({\bf Y}({\bf a}))$ for any value assignment
${\bf a} \in {\mathfrak X}_{\bf A}$ can be
identified under the SWM as a functional of the observed
distribution $p({\bf V})$ using the \emph{extended g-formula}
\cite{robins04effects}, given by
\begin{align}
p({\bf Y}({\bf a}) = {\bf v}_{\bf Y}) =
	\sum_{{\bf v}_{{\bf V} \setminus {\bf Y}}}
	\prod_{V \in {\bf V}} p({\bf v}_V \mid
		{\bf a}_{\pa_{\cal G}(V) \cap {\bf A}},
		{\bf v}_{\pa_{\cal G}(V) \setminus {\bf A}})
\label{eqn:g-formula}
\end{align}
where ${\bf v} \in {\mathfrak X}_{\bf V}$.  A recent proof of this
appears in \cite{thomas13swig}.  Special cases of (\ref{eqn:g-formula})
where ${\bf A}$ and ${\bf Y}$ are disjoint are known as the \emph{g-formula}
\cite{robins86new}, the \emph{manipulated distribution}
\cite{spirtes01causation}, or the \emph{truncated factorization}
\cite{pearl09causality}.
Because the MWM is a causal submodel of the SWM, (\ref{eqn:g-formula}) also
holds under the MWM.

\subsection{Total Effects as Responses to Node Interventions}
\label{sec:total-as-node}

Node interventions are used to represent causal effects of treatments as a
contrast of potential outcome responses to different treatment assignments.
By considering an intervention we remove the impact of confounding via
assignment policy.  For example, consider the simple causal graph shown in Fig.
\ref{fig:triangle} (a), representing an observational study with a
single application of one of two treatments $m, m'$.  Variable $M$ is
assigned to either $m$ or $m'$ based on (observed) patient health
status ($A,W$), and survival $Y$ is measured.  Doctors follow a known policy
$p(M \mid A,W)$ in assigning $M$ where sicker patients are more likely to get
$m$.  Note that $p(\text{alive} \mid m) < p(\text{alive} \mid m')$ may
hold simply due to the assignment policy in the study which introduces
confounding by health status, even if $m$ is a better drug.

One appropriate contrast that adjusts for the influence of confounding by
health status on the effect of interest
can be expressed via node interventions, and is known as the average causal
effect (ACE): $\mathbb{E}[Y(m)] - \mathbb{E}[Y(m')]$.
This contrast can be computed from the distribution $p(Y(m))$ for all
$m \in {\mathfrak X}_M$, which is equal, under (\ref{eqn:g-formula}), to
\[
p(Y(m)) = \!\!\! \sum_{w,a,m'} p(Y \mid m,a,w) p(m' \mid a,w) p(a,w) =
	\sum_{w,a} p(Y \mid m,a,w) p(a,w).
\]
This recovers the well-known 
back-door
formula \cite{pearl09causality}.

Consider now a more complex example corresponding to
the following problem from HIV research.
In a longitudinal study, HIV patients were put on an antiretroviral drug
regimen, where the specific level of drug exposure over time was controlled by
a known policy, which was based on covariates observed for each patient.
However, the outcome of the study has been disappointing.
The question is whether this was due to the drug itself performing poorly,
or whether patient's adherence was poor.  Consider a causal graph
representing two time slices of this longitudinal study.  To avoid cluttering
the figure with too many edges, we represent the causal graph schematically
by its transitive reduction with respect to blue edges, shown
in Fig. \ref{fig:triangle} (b).
That is, the true graph ${\cal G}^*$ contains a blue arrow between any
pair of nodes $A,B$ connected by a blue directed path in
Fig. \ref{fig:triangle} (b) (and inherits all red edges as well).

Here $C_0$ is a vector of observed baseline confounders,
$A_1,A_2$ are exposures over
time, $W_1,W_2$ are drug toxicity levels at each exposure time, $C_1,C_2$
are adherence levels at each time, $Y_1,Y_2$ are outcomes, and $U$ is an
unobserved confounder.  Both red and blue arrows represent direct causation.
In general, a reasonable causal graph will contain unobserved common
causes of most vertices, but in this example we assume adherence $C_1,C_2$,
and treatments $A_1,A_2$ are only \emph{directly} affected by the observed
variables in the past, such as the toxicity level of the drug, and not by $U$.
These assumptions are represented graphically by the
absence of red edges from $U$ to $A_1,A_2,C_1,C_2$.

We first consider the total effect of the two exposures on outcome $Y_2$,
formalized as the two-exposure version of ACE.  We consider more complex
effects involving mediation by adherence in subsequent sections.
The ACE contrast is defined with respect to active treatment levels, which we
denote $a_1,a_2$, and baseline treatment levels, which we denote $a'_1,a'_2$.
In our case, the contrast is equal to
$\text{ACE} \equiv
\mathbb{E}[Y_2(a_1, a_2)] - \mathbb{E}[Y_2(a'_1, a'_2)]$.
If we were able to randomize treatment assignment to $A_1,A_2$, we could
evaluate the ACE directly from experimental data.  However, our data comes from
an observational longitudinal study, and therefore we must properly adjust
for observed confounders of the exposures.  Robins \cite{robins86new} noted
that in cases like these, assuming the underlying SWM represented by
our graph is correct, we can get a bias-free estimand of the ACE from
observational data using
the \emph{g-computation algorithm}, which in this case gives
\begin{align*}
\text{ACE} =
\sum_{y_1, c_1, w_1, c_0} &
\mathbb{E}[Y_2 \mid a_2, y_1, c_1, w_1, a_1, c_0]
p(y_1, c_1, w_1 \mid a_1, c_0)
p(c_0) -\\
\sum_{y_1, c_1, w_1, c_0} &
\mathbb{E}[Y_2 \mid a'_2, y_1, c_1, w_1, a'_1, c_0]
p(y_1, c_1, w_1 \mid a'_1, c_0)
p(c_0)
\end{align*}
This is, yet again, a special case of (\ref{eqn:g-formula}).
This estimand can be estimated via either the parametric g-formula
\cite{robins87graphical}, inverse weighting methods \cite{robins00marginal}, or
doubly robust methods \cite{robins94estimation}.

In the following section, we introduce intervention types that generalize
node interventions, and consider other types of causal effects which may be
represented as responses to such intervention types.

\section{Edge and Path Interventions}
\label{sec:edge-path}

We consider two additional types of interventions defined on graphical
features, edge and path interventions, and define responses to these
interventions using recursive substitution in a natural way.
As we shall see, responses to path interventions include many targets of
interest in causal inference, including effects of treatment on the treated,
mediated effects, and even novel effects that combine features of both.

\subsection{Edge Interventions}

For a set of edges ${\bm\alpha}$ in a DAG ${\cal G}$, define
${\mathfrak X}_{\bm\alpha} \equiv
{\mathfrak X}_{\sop_{\cal G}({\bm\alpha})}$.
In other words, ${\mathfrak X}_{\bm\alpha}$ is a Cartesian product of the
state spaces of source variables of all directed edges in
${\bm\alpha}$.

The state space of a given vertex in ${\cal G}$ may occur
\emph{multiple times} in ${\mathfrak X}_{\bm\alpha}$ if multiple edges in
${\bm\alpha}$ share the same source vertex.  We denote members of
${\mathfrak X}_{\bm\alpha}$ by lowercase Frankfurt font:
${\mathfrak a} \in {\mathfrak X}_{\bm\alpha}$.  We do so to emphasize that
elements of ${\mathfrak X}_{\bm\alpha}$
may contain multiple \emph{conflicting}
value assignments to the same random variable, unlike elements of
${\mathfrak X}_{\bf A}$.  For example, consider the graph in
Fig. \ref{fig:triangle} (a), where ${\mathfrak X}_{A} = \{ 0, 1 \}$.
Then if ${\bm\alpha} = \{ (AM)_{\to}, (AY)_{\to} \}$, a valid element
${\mathfrak a}$ of ${\mathfrak X}_{\bm\alpha}$
associates $0$ with the variable associated with the parent vertex $A$ of
$(AM)_{\to}$ and $1$ with the variable associated with the
parent vertex $A$ of $(AY)_{\to}$.  Unlike elements of
${\mathfrak X}_{\bf A}$,
it is not immediately clear what set of edges ${\mathfrak a}$ is referring to,
so we will subscript the set of edges if necessary, like so:
${\mathfrak a}_{\bm\alpha}$.

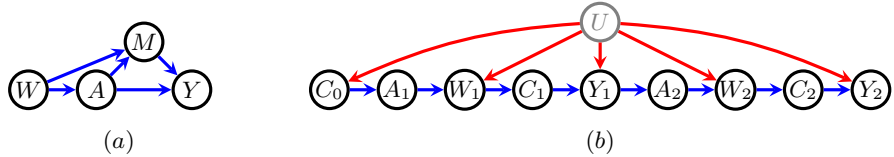
\begin{figure}
\begin{center}
\begin{tikzpicture}[>=stealth, node distance=0.9cm]
    \tikzstyle{format} = [draw, very thick, circle, minimum size=5.0mm,
	inner sep=0pt]

	\begin{scope}
		\path[->, very thick]
			node[format] (w) {$W$}
			node[format, right of=w] (a) {$A$}
			node[format, above right of=a] (m) {$M$}
			node[format, below right of=m] (y) {$Y$}

			(w) edge[blue] (a)
			(a) edge[blue] (y)
			(a) edge[blue] (m)
			(m) edge[blue] (y)
			(w) edge[blue] (m)

			node[below of=a, yshift=0.2cm, xshift=0.3cm] (l) {$(a)$}
		;
	\end{scope}
	\begin{scope}[xshift=4.0cm]
		\path[->, very thick]
			node[format] (c0) {$C_0$}
			node[format, right of=c0] (a1) {$A_1$}
			node[format, right of=a1] (w1) {$W_1$}
			node[format, right of=w1] (c1) {$C_1$}
			node[format, right of=c1] (y1) {$Y_1$}

			node[format, right of=y1] (a2) {$A_2$}
			node[format, right of=a2] (w2) {$W_2$}
			node[format, right of=w2] (c2) {$C_2$}
			node[format, right of=c2] (y2) {$Y_2$}

			node[format, above of=y1, gray] (u) {$U$}

				(c0) edge[blue] (a1)
				(a1) edge[blue] (w1)
				(w1) edge[blue] (c1)
				(c1) edge[blue] (y1)
				(y1) edge[blue] (a2)
				(a2) edge[blue] (w2)
				(w2) edge[blue] (c2)
				(c2) edge[blue] (y2)

				(u) edge[->, bend right=10, red] (c0)
				(u) edge[->, red] (w1)
				(u) edge[->, red] (y1)
				(u) edge[->, red] (w2)
				(u) edge[->, bend left=10, red] (y2)

			node[below of=y1, yshift=0.2cm] (l) {$(b)$}
			;
	\end{scope}

\end{tikzpicture}
\end{center}
\caption{(a) A simple causal graph.
(b) The transitive closure with respect to blue arrows
of this graph is a causal graph representing two time slices of
a longitudinal study in HIV research.
}
\label{fig:triangle}
\end{figure}

We call a forced assignment of variables corresponding to source vertices
of edges from ${\bm\alpha}$ to an element of ${\mathfrak X}_{\bm\alpha}$ an
\emph{edge intervention}.  An edge intervention which assigns ${\bm\alpha}$
to an element ${\mathfrak a}_{\bm\alpha} \in {\mathfrak X}_{\bm\alpha}$ will
be denoted by $\eta_{{\mathfrak a}_{\bm\alpha}}$.
As with elements of ${\mathfrak X}_{\bf A}$, we
denote a restriction of ${\mathfrak a}$ by a set subscript.  That is, if
${\mathfrak a}_{\bm\alpha} \in {\mathfrak X}_{\bm\alpha}$, and
${\bm\beta} \subseteq {\bm\alpha}$, then ${\mathfrak a}_{\bm\beta}$ is
a restriction of ${\mathfrak a}$ to variables corresponding to source vertices
of ${\bm\beta}$.

We define responses of outcomes to edge
interventions in the natural way using recursive substitution, the
potential outcomes of the form $V({\bf v}_{\pa_{\cal G}(V)})$, and a joint
distribution over these potential outcomes.
For every $V \in {\bf V}$, a set of edges ${\bm\alpha}$ in a DAG ${\cal G}$,
and an element ${\mathfrak a}_{\bm\alpha} \in {\mathfrak X}_{\bm\alpha}$,
we define the response of $V$ to $\eta_{{\mathfrak a}_{\bm\alpha}}$ as
\begin{align}
V({\mathfrak a}_{\bm\alpha}) \equiv
V({\mathfrak a}_{\{ (*V)_{\to} \in {\bm\alpha} \}},
\{ \pa^{\overline{\bm\alpha}}_{\cal G}(V) \}({\mathfrak a}_{\bm\alpha}))
\label{eqn:edge-rec-sub}
\end{align}
where $\pa^{\overline{\bm\alpha}}_{\cal G}(V) \equiv
	\{ A \in \pa_{\cal G}(V) \mid
	(AV)_{\to} \not\in {\bm\alpha} \}$.

In words, this states that the response of $V$ to
$\eta_{{\mathfrak a}_{\bm\alpha}}$, where
${\mathfrak a}_{\bm\alpha} \in {\mathfrak X}_{\bm\alpha}$ is
defined as the potential outcome where all parents of $V$ along edges in
${\bm\alpha}$ are assigned an appropriate value from
${\mathfrak a}_{\bm\alpha}$,
and all other parents are assigned whatever value they would have attained
under an edge intervention $\eta_{{\mathfrak a}_{\bm\alpha}}$ (these are defined
recursively, and the definition terminates because of the lack of directed
cycles in ${\cal G}$).

As before, given a set ${\bf Y} = \{ Y_1, \ldots, Y_k \}$ of random variables,
we denote $\{ Y_1({\mathfrak a}_{\bm\alpha}), \ldots,
Y_k({\mathfrak a}_{\bm\alpha}) \}$ by
${\bf Y}({\mathfrak a}_{\bm\alpha})$ or
$\{ {\bf Y} \}({\mathfrak a}_{\bm\alpha})$.

\subsection{Direct and Indirect Effects as Responses to Edge Interventions}
\label{sec:dir-indir-edge}

Just as responses to node interventions can be used to represent total causal
effects, so can responses to edge interventions be used to represent
direct and indirect effects.  Consider again Fig. \ref{fig:triangle} (a),
but now assume $A$ is the treatment (one of two drugs $a,a'$),
$Y$ is the outcome (survival), and $M$ is a dangerous side effect that
mediates some of the effect of $A$ on $Y$.

We may be interested in how much of the total effect, as formalized via the ACE
contrast $\mathbb{E}[Y(a)] - \mathbb{E}[Y(a')]$, can be attributed to the
direct effect of the drugs on $Y$, and how much to the mediated effect via the
side effect $M$.  To formalize this, we want to consider how $Y$ varies if
we can set treatments separately for the purposes of the direct causal pathway
represented by $(AY)_{\to}$ and the pathway mediated by $M$, represented by
$(AM)_{\to}$.  This is precisely what edge interventions allow us to do.
Consider $\eta_{\mathfrak a}$ that sets
$(AM)_{\to}$ to $a$ and $(AY)_{\to}$ to $a'$.  Then (\ref{eqn:edge-rec-sub})
implies $Y({\mathfrak a}) = Y(a',M(a))$.
We can use this type of response to define the direct effect as the contrast
$\mathbb{E}(Y(a)) - \mathbb{E}[Y(a',M(a))]$, and the indirect effect as
the contrast $\mathbb{E}[Y(a',M(a))] - \mathbb{E}[Y(a')]$.
Note that the ACE is a sum of the direct and indirect effect contrasts above.

The idea of using nested responses like $Y(a',M(a))$ to represent
direct and indirect effects for mediation analysis appears in
\cite{robins92effects}, and is discussed in
the context of graphical causal models in \cite{pearl01direct}.  Our
contribution is to aid interpretability of such nested responses by viewing
them as responses to interventions associated with edges, graphical features
intuitively associated with effects we are trying to formalize.


Just as it is good practice to only
discuss node interventions in settings where it is possible, at least in
principle, to assign treatment by fiat, so it is good practice to only
discuss edge interventions in settings where it is possible, at least in
principle, to conceive of assigning only those components of the overall
treatment that influences a particular direct consequence.  For instance,
if smoking affects cardiovascular disease only by means of nicotine content,
then we might simulate the absence of smoking, but only for the purposes of
cardiovascular disease, by assigning the ``treatment''
of nicotine-free cigarettes.
In this paper, we leave the issues of applicability of edge interventions and
mediation analysis in particular settings aside \cite{robins10alternative},
and consider, in subsequent sections, questions of identification and
the form of resulting functionals.

\subsection{Path Interventions}

We are going to define responses to path interventions, which associate a
set of directed paths with values of sources of every path in the set.
A response to a path intervention will behave as if the source of a path were
set to a particular value, \emph{but only for the purposes of a particular
outgoing directed path}.  This behavior generalizes the behavior of edge
interventions, where vertices may behave differently with respect to different
outgoing edges.  Path interventions serve as a very general, graphical
representation of counterfactual quantities associated with causal pathways
that generalizes both edge and path interventions.
The supplementary materials \cite{shpitser15nep}
contain our rationale for the use of path
interventions versus simpler or more algebraic approaches to representing
counterfactuals of interest.

To make sure we end up with well-defined responses, we insist on a property
for sets of directed paths called \emph{properness}.
A set of directed paths ${\bm\alpha}$ in a DAG ${\cal G}$ is called
\emph{proper} if no path in ${\bm\alpha}$ is a prefix subpath of
another path in ${\bm\alpha}$.  A set consisting of a single path is
always proper, as is a set of length 1 paths (e.g. a set of edges).
In the remainder of the paper, when we say ``a set of paths ${\bm\alpha}$,''
we mean a proper set of directed paths.

For a set of paths ${\bm\alpha}$ in a DAG ${\cal G}$, define
${\mathfrak X}_{\bm\alpha} \equiv
	{\mathfrak X}_{\sop_{\cal G}({\bm\alpha})}$.
In other words, ${\mathfrak X}_{\bm\alpha}$ is a Cartesian product of the
state spaces of source variables of all directed paths in ${\bm\alpha}$.
Since sets of paths clearly generalize sets of edges, the same issue occurs
where a single vertex in ${\cal G}$ may occur multiple times in
${\mathfrak X}_{\bm\alpha}$.  As before, to emphasize this, we will denote
elements of ${\mathfrak X}_{\bm\alpha}$ by lowercase Frankfurt font:
${\mathfrak a}$, possibly indexed by a path set subscript:
${\mathfrak a}_{\bm\alpha}$.

We denote a forced assignment of variables corresponding to source vertices
of paths from ${\bm\alpha}$ to an element of ${\mathfrak X}_{\bm\alpha}$ as a
\emph{path intervention}.  A path intervention which assigns ${\bm\alpha}$
to an element ${\mathfrak a}_{\bm\alpha} \in {\mathfrak X}_{\bm\alpha}$ will
denoted by $\pi_{{\mathfrak a}_{\bm\alpha}}$.
As with elements of ${\mathfrak X}_{\bf A}$, we
denote a restriction of ${\mathfrak a}$ by a set subscript.  That is, if
${\mathfrak a}_{\bm\alpha} \in {\mathfrak X}_{\bm\alpha}$, and
${\bm\beta} \subseteq {\bm\alpha}$, then ${\mathfrak a}_{\bm\beta}$ is
a restriction of ${\mathfrak a}$ to variables corresponding to source vertices
of ${\bm\beta}$.

As was the case with node and edge interventions, our definition of path
interventions will be inductive.  To get the induction to work,
we need to consider how treatments affect the response via pathways that end
in a particular edge.  We use the following definition to formalize this.
Given a set of paths ${\bm\alpha}$ in a DAG ${\cal G}$, and
an edge $(WY)_{\to}$, define a \emph{funnel operator}
$\lhd_{(WY)_{\to}}$ which maps from ${\bm\alpha}$ to the set of paths
$\lhd_{(WY)_{\to}}({\bm\alpha})$ obtained from ${\bm\alpha}$ by
replacing any path of the form $(A,\ldots,W,Y)_{\to}$ by $(A,\ldots,W)_{\to}$,
by removing all paths containing $W$ but no suffix $(WY)_{\to}$,
and keeping all other paths intact.

\begin{lem}
If ${\bm\alpha}$ is proper, then for any edge $(WY)_{\to}$, so is
$\lhd_{(WY)_{\to}}({\bm\alpha})$.
\label{lem:proper-funnel}
\end{lem}
Given a path intervention $\pi$ that assigns ${\bm\alpha}$ to
${\mathfrak a}_{\bm\alpha}$, and a funnel operator $\lhd_{(WY)_{\to}}$,
we consider \emph{funneled} path interventions on
$\lhd_{(WY)_{\to}}({\bm\alpha})$.  For every
$\alpha$ such that $\lhd_{(WY)_{\to}}(\alpha) = \alpha$, the funneled
path intervention assigns $\alpha$ to ${\mathfrak a}_{\alpha}$, that is it
keeps the same value assignment as the original path intervention.
For the path $\alpha \equiv (A\ldots W,Y)_{\to}$
the funneled path intervention assigns
$\lhd_{(WY)_{\to}}(\alpha)$ to ${\mathfrak a}_{(A\ldots WY)_{\to}}$, that is
assigns the value given by the original intervention to $(A\ldots WY)_{\to}$.
We denote such an assignment by
${\mathfrak a}_{\lhd_{(WY)_{\to}}({\bm\alpha})}$.

Our insistence on ${\bm\alpha}$ being proper, together with
Lemma \ref{lem:proper-funnel}, means that there is never any ambiguity in
defining the funneled path intervention.  That is, it is
never the case that two distinct paths in ${\bm\alpha}$ are of the form
$(A\ldots W)_{\to}$ and $(A\ldots WY)_{\to}$.  If such a pair of paths were
allowed, the difficulty would then be that these paths can both
reasonably be claimed to represent an effect of setting $A$ along the path
$(A\ldots WY)_{\to}$, while potentially disagreeing on what that setting is.

We are now ready to define responses to path interventions.  For every
$V \in {\bf V}$, a proper set of directed paths ${\bm\alpha}$ in a DAG
${\cal G}$, and an element
${\mathfrak a}_{\bm\alpha} \in {\mathfrak X}_{\bm\alpha}$, we define the
response of $V$ to $\pi_{{\mathfrak a}_{\bm\alpha}}$ as
\begin{align}
V({\mathfrak a}_{\bm\alpha}) \equiv
	V({\mathfrak a}_{(*V)_{\to} \in {\bm\alpha}},
	\{ W( {\mathfrak a}_{\lhd_{(WY)_{\to}}({\bm\alpha})} ) \mid
	W \in \pa_{\cal G}^{\overline{\bm\alpha}}(V) \})
\label{eqn:path-rec-sub}
\end{align}
where
$\pa_{\cal G}^{\overline{\bm\alpha}}(V) \equiv \{ W \in \pa_{\cal G}(V) \mid
	(WV)_{\to} \not\in {\bm\alpha} \}$.

In words, this states that the response of $V$ to
$\pi_{{\mathfrak a}_{\bm\alpha}}$, where ${\mathfrak a}_{\bm\alpha}
\in {\mathfrak X}_{\bm\alpha}$ is defined as the potential outcome where all
parents of $V$ along edges which are (length 1) paths in ${\bm\alpha}$ are
assigned an appropriate value from ${\mathfrak a}_{\bm\alpha}$, and
all other parents $W$ are assigned whatever value they would have attained under
the funneled path intervention associated with a funnel operator
for the edge between that parent $W$ and $V$.
Note that the definition is inductive for such parents, with the result of
applying a funnel operator serving as the new set of paths.
Lemma \ref{lem:proper-funnel} ensures that properness propagates to this set,
and thus the overall response is well-defined.

For example, if $\pi_{\mathfrak a}$ assigns $w$ to $(WAMY)_{\to}$ in
Fig. \ref{fig:triangle} (a), then $Y({\mathfrak a})$ is defined by
(\ref{eqn:path-rec-sub}) to equal $Y(M(A(w)),A)$.
We will use a notational shorthand for responses to path interventions,
where rather than listing nested responses in parentheses after the response,
we list the paths with the source node replaced by the intervened on value.
For example, we write $Y({\mathfrak a}) = Y(M(A(w)),A)$ above as
$Y((wAMY)_{\to})$.  We use the same shorthand for responses to
edge interventions.


As before, given a set ${\bf Y} = \{ Y_1, \ldots, Y_k \}$ of random variables,
we denote $\{ Y_1({\mathfrak a}_{\bm\alpha}), \ldots,
Y_k({\mathfrak a}_{\bm\alpha}) \}$ by
${\bf Y}({\mathfrak a}_{\bm\alpha})$ or
$\{ {\bf Y} \}({\mathfrak a}_{\bm\alpha})$.

\subsection{Responses to Path Interventions to Natural Values}
\label{sec:natural}

So far we have defined path interventions as a mapping from 
a proper set of directed paths ${\bm\alpha}$ to values in
${\mathfrak X}_{\bm\alpha}$.  However, we might be interested in considering
responses to interventions that assign a variable not to a specific
constant value, but to a value the variable would have attained under a no
intervention regime.  For instance, this might happen if the baseline exposure
is one received by the general population, not a specific exposure level
assigned by the experimenter, or if the effect of multiple treatments on the
treated is of interest.  In the context of node interventions, this situation
was discussed in \cite{population08hubbard}.  In order for
responses to path interventions to include this case, we must extend the
definition of path interventions to include intervening to \emph{natural}
values, that is values attained by variables under no interventions.
Allowing arbitrary variables to be set
to natural values may lead to identification difficulties even in very simple
cases.  Consider the following response to a node intervention in the MWM
given by Fig. \ref{fig:triangle} (a), $\{ A, Y \}(A,w)$.  In words, this is
the joint response of $A$ and $Y$ to an intervention where $W$ is set to value
$w$, and $A$ is set to the natural value it attains under no interventions.
The definition of responses to node interventions via recursive substitution
shows that $\{ A, Y \}(A,w) = Y(A),A(w)$.  However, the distribution
$p(A,A(w))$ is not identified under the MWM for the graph in
Fig. \ref{fig:triangle} (a), see Lemma \ref{lem:chen}, and thus neither is
the joint response in question.

To avoid this difficulty, we consider only a special subset of path
interventions containing settings on natural values.  This special subset can
safely be rephrased in such a way that only interventions on constants remain
explicit.  To define this special subset, we need a few preliminary definitions.

For a node $A$, and a directed path (or an edge) $\alpha$
with source $A$, define the \emph{extended state space} as follows
${\mathfrak X}^*_A \equiv {\mathfrak X}_A \cup \{ A \}$, and
${\mathfrak X}^*_{\alpha} \equiv {\mathfrak X}_{\alpha} \cup \{ A \}$.
We define the extended state space for sets of nodes, edges, and paths
disjunctively as before.
An intervention on an extended state space is allowed on either any constant
value, or on the ``natural value.''

Given a set of paths ${\bm\alpha}$ and a response set
${\bf Y}$, we call a directed path $\alpha$ \emph{relevant} for ${\bf Y}$
given ${\bm\alpha}$ if $\alpha = (A \ldots Y)_{\to}$, where $Y \in {\bf Y}$,
and no path in ${\bm\alpha}$ is a subpath of $\alpha$ except possibly a prefix
of $\alpha$.  We denote the set of all relevant paths for ${\bf Y}$ given
${\bm\alpha}$ in ${\cal G}$ by $\rel_{\cal G}({\bf Y}\mid{\bm\alpha})$.

Paths relevant for ${\bf Y}$ given ${\bm\alpha}$ are those paths consisting
of sequences of intermediate responses that arise in the inductive definition
(\ref{eqn:path-rec-sub}).  For example, assume we are interested in
the singleton response set $\{ Y \}$ and a singleton path set
$\{ (WAMY)_{\to} \}$ in Fig. \ref{fig:triangle} (a).
Then defining $Y((wAMY)_{\to})$ for a particular
$w$ via (\ref{eqn:path-rec-sub}) entails defining intermediate responses
$M((wAM)_{\to})$ and $A((wA)_{\to})$.  The sequence of vertices $(A,M,Y)$
are all linked by directed edges by (\ref{eqn:path-rec-sub}), and $(AMY)_{\to}$
is relevant for $\{ Y \}$ given $\{ (WAMY)_{\to} \}$.
Similarly, $(WAMY)_{\to}$ and $(WAY)_{\to}$ are relevant for
$\{ Y \}$ given $\{ (WAMY)_{\to} \}$.


We now give two useful results about relevant paths.

\begin{lem}
If $\alpha \in \rel_{\cal G}({\bf Y} \mid {\bm\alpha})$, then
$\beta \in \rel_{\cal G}({\bf Y} \mid {\bm\alpha})$
for any suffix subpath $\beta$ of $\alpha$.
\label{lem:suffix-rel}
\end{lem}

\begin{lem}
If 
${\bm\beta} \subseteq {\bm\alpha}$, then for any ${\bf Y}$,
$\rel_{\cal G}({\bf Y}\mid{\bm\alpha}) \subseteq
\rel_{\cal G}({\bf Y}\mid{\bm\beta})$.
\label{lem:subset-rel}
\end{lem}

A set of interventions may not all have an effect on a response, due to
constraints of the model.
For instance, since $Y(a,m,w) \neq Y(a',m,w)$ but $Y(a,m,w) = Y(a,m,w')$ for
any $m,a,w,a',w'$ in Fig. \ref{fig:triangle} (a), $A$ has an effect on $Y$, but
$W$ does not, given that we also intervene on $A$ and $M$.
We extend this notion to path interventions, and call those
paths with sources that actually have an effect on the response, given
interventions on other paths, \emph{live}.
More precisely, given a proper set of paths ${\bm\alpha}$ and a response set
${\bf Y}$, we call a path $\alpha \in {\bm\alpha}$ \emph{live} for ${\bf Y}$
given ${\bm\alpha}$ if there is an element of
$\rel_{\cal G}({\bf Y}\mid{\bm\alpha})$ containing $\alpha$ as a prefix.

Consider the maximal subset of ${\bm\alpha}$ consisting of paths in
${\bm\alpha}$ live for ${\bf Y}$ given ${\bm\alpha}$, or
${\bm\alpha}_{\bf Y} \equiv \{ \alpha \in {\bm\alpha} \mid \alpha
	\text{ live for } {\bf Y} \text{ given } {\bm\alpha} \}$.
We say a set of directed paths ${\bm\alpha}$ is live for ${\bf Y}$
if ${\bm\alpha} = {\bm\alpha}_{\bf Y}$.  When discussing path interventions,
we can always restrict our attention to sets of paths live for ${\bf Y}$
without loss of generality, due to the following result.

\begin{lem}
For any ${\bf Y}$ and ${\bm\alpha}$ proper for ${\bf Y}$,
$\rel_{\cal G}({\bf Y}\mid{\bm\alpha}) =
\rel_{\cal G}({\bf Y}\mid{\bm\alpha}_{\bf Y})$,
$({\bm\alpha}_{\bf Y})_{\bf Y} = {\bm\alpha}_{\bf Y}$, and in addition,
for any ${\mathfrak a}_{\bm\alpha}$,
$p({\bf Y}({\mathfrak a}_{\bm\alpha})) =
p({\bf Y}({\mathfrak a}_{{\bm\alpha}_{\bf Y}}))$.
\label{lem:path-live}
\end{lem}

We now show that we can either ignore interventions to natural values in a
response to a path intervention, or the response is not identified under the
MWM.  The set of paths for which the former is true for the response ${\bf Y}$
will be called \emph{natural} for ${\bf Y}$.
Due to this result, we do not need to consider interventions to natural
values explicitly.

\begin{dfn}
Let ${\bm\alpha}$ be live for ${\bf Y}$.
Let $\pi_{{\mathfrak a}_{\bm\alpha}}$ be a path intervention in ${\cal G}$
where a subset ${\bm\alpha}^* \subseteq {\bm\alpha}$ is assigned constant
values, and ${\bm\alpha} \setminus {\bm\alpha}^*$ is assigned natural values.
Then if no element of $\rel_{\cal G}({\bf Y}\mid{\bm\alpha}^*)$
with a prefix subpath in ${\bm\alpha}^*$ contains a subpath in
${\bm\alpha} \setminus {\bm\alpha}^*$, we say $\pi$ is \emph{natural} for
${\bf Y}$.
\end{dfn}

\begin{lem}
Let $\pi_{{\mathfrak a}_{\bm\alpha}}$ be a path intervention
natural for ${\bf Y}$, and ${\bm\alpha}^* \subseteq {\bm\alpha}$ is
all paths assigned constant values by $\pi$.
Then $p({\bf Y}({\mathfrak a}_{\bm\alpha})) =
p({\bf Y}({\mathfrak a}_{{\bm\alpha}^*}))$.
\label{lem:reduc-nat}
\end{lem}

\begin{lem}
If $\pi_{{\mathfrak a}_{\bm\alpha}}$ is not natural for ${\bf Y}$ in ${\cal G}$,
then $p({\bf Y}({\mathfrak a}_{\bm\alpha}))$ is not identified under the
MWM for ${\cal G}$.
\label{lem:not-nat-not-id}
\end{lem}

Lemma \ref{lem:reduc-nat} does not guarantee that a response to a natural
path intervention is identifiable, merely that it can be expressed as a response
to an intervention only setting to constant values.

\section{Causal Inference Targets as Responses to Path Interventions}
\label{sec:targets}

In this section we consider how a number of targets of interest
in causal inference, including novel targets not previously considered in the
literature, may be expressed as responses to path interventions.

We use as our running example the two time point fragment of a
longitudinal study in HIV research, described in Section
\ref{sec:total-as-node}.
We consider path-specific effects that arise in mediation analysis,
and effects of treatment on the multiply treated, which
are of interest in tort cases (since these are effects of the exposure
on those actually exposed), and in epidemiology if natural exposure levels
carry information about the causal effect of the exposure.
It is not straightforward to see whether these types of effects
are identifiable, and under what model, nor is it obvious whether there is
a single unifying principle which governs identification for these effects.   

By translating the effect types above into responses to path interventions,
we show that such responses form a very general class of causal inference
targets.  Thus, the advantage of path interventions is that we can
use them to give a single characterization for a wide variety of targets of
interest at once.  The close relationship between effects of treatment on the
treated and mediated effects hinted by their common generalization as
responses to path interventions is currently not widely known.

We will define a special set of directed paths important for our translation
scheme.  Given a treatment set ${\bf A}$ and an outcome set ${\bf Y}$ (that
possibly intersect) in a DAG
${\cal G}$, define the set ${\bm\alpha}_{{\bf A},{\bf Y},{\cal G}}$ to be the
set of all directed paths with a source in ${\bf A}$, a sink in
${\bf A} \cup {\bf Y}$ and which do not intersect ${\bf A} \cup {\bf Y}$ except
at the source and sink.
Since ${\bf A}$ and ${\bf Y}$ are allowed to intersect, the names ``treatment''
and ``outcome'' are slightly misleading.  We allow the intersection to admit
cases such as effect of treatment on the treated (ETT) where some treatments
are also treated as responses for the purposes of certain paths.

\begin{lem}
${\bm\alpha}_{{\bf A},{\bf Y},{\cal G}}$ is always proper.
\label{lem:always-proper}
\end{lem}

\subsection{Effects of Treatment on the Treated}

We consider an effect on the mean difference scale where we condition
on the naturally observed treatment levels.
This is known as the effect of treatment
on the treated (ETT), and in our two time point HIV example,
it is defined as follows
\[
\text{ETT} \equiv
\mathbb{E}[Y(a_1,a_2) \mid a_1, a_2] - \mathbb{E}[Y(a'_1,a'_2) \mid a_1, a_2].
\]

This contrast is often of interest to epidemiologists.  It also
arises in cases where interventions are functions of the natural value of
the exposure.  For example, we may be interested in outcome for people who
were encouraged to exercise for 30 more minutes than they normally would, which
is a random variable of the form $Y(A+30) \equiv Y(a+30) \mid A=a$.  These
types of interventions are discussed in \cite{identification14jessica},
in particular sufficient conditions for identification under the SWM, in
terms of the extended g-formula (\ref{eqn:g-formula}) are given
there and in \cite{thomas13swig}.

Assume $A_1$ is a binary variable (only two treatment levels).  If we
consider, instead,
the ETT with respect to only the exposure $A_1$, we obtain the
following derivation for the second term in the contrast
\begin{align*}
p(Y_2(a'_1) \mid a_1) = \frac{p(Y_2(a'_1), a_1)}{p(a_1)} =
\frac{p(Y_2(a'_1)) - p(Y_2(a'_1), a'_1)}{p(a_1)},
\end{align*}
where the first identity is by definition, and the second by the binary
treatment assumption.
Since consistency implies $p(Y_2(a_1), a_1) = p(Y_2, a_1)$ for any value $a_1$,
the ETT for a single binary exposure $A_1$ can be identified if
$p(Y_2(a_1))$ is identified.

However, if the exposure is not binary, or if there are multiple exposures, as
in our example, we cannot use the same algebraic trick to obtain identification,
and we must proceed by exploiting additional assumptions in our causal model.

In our case, the first conditional mean in the contrast can be readily
identified via consistency:
$\mathbb{E}[Y(a_1, a_2) \mid a_1, a_2] = \mathbb{E}[Y \mid a_1, a_2]$.
However, the second conditional mean presents a problem, because it contains
a conflict between the naturally observed exposures, and the assigned
exposures.  Here we show how to represent the underlying joint distribution
over potential outcomes, $p(Y_2(a_1, a_2), A_1, A_2)$, in terms of path
interventions, and then attack the identification problem for \emph{all}
responses to path interventions, which would then include the problematic
second term of the ETT.

We consider all directed paths from $A_2$ to $Y_2$, which we assign a value
$a_2$, all directed paths from $A_1$ to $Y_2$ not through $A_2$, which we
assign a value $a_1$, and all directed paths from $A_1$ to $A_2$, which we
assign the natural value of $A_1$.  Note that this set of paths is simply
${\bm\alpha}_{\{ A_1,A_2\},\{Y_2\},{\cal G}}$ for ${\cal G}$ that is the
transitive closure with respect to blue edges of the graph in Fig.
\ref{fig:triangle} (b), and thus is proper by Lemma \ref{lem:always-proper}.
We then
consider the response of $A_1, A_2, Y_2$ to the path intervention so defined,
or $\{ A_1, A_2, Y_2 \}({\mathfrak a}_{\bm\alpha})$.  By our definition,
all paths set to a value ancestral for $A_1,A_2$ are set to natural values.
Thus, $\{ A_1, A_2 \}({\mathfrak a}_{\bm\alpha})$ is defined in terms of natural
values of its direct causal parents, or as
$A_1(C_0) = A_1$ and $A_2(Y_1,C_1,W_1,A_1,C_0) = A_2$.

Finally, we consider all paths ancestral for $Y_2$.
Since $A_1$ and $A_2$ are parents
of $Y_2$ in ${\cal G}^*$, the single edge paths $(A_1Y_2)_{\to}$ and
$(A_2Y_2)_{\to}$ are in our set, thus we substitute $a_1$ and $a_2$ into
the potential outcome answer.  Furthermore, for other parents of $Y_2$,
namely $C_0,U,W_1,C_1,Y_1,W_2$ and $C_2$, we consider an appropriate
set derived from ${\bm\alpha}$.  For example, for the node $W_2$, we replace
the path $A_2 \to W_2 \to Y_2$ by a path $A_2 \to W_2$ (while keeping the
assignment $a_2$).  We proceed in this way recursively until we obtain the
response for $Y_2$, which is
{\small
\[
Y_2(a_1,a_2,U,C_0, W_1(a_1, \ldots), C_1(a_1, \ldots),
	Y_1(a_1, \ldots), W_2(a_1,a_2, \ldots), C_2(a_1, a_2, \ldots)),
\]
}
where $\ldots$ is a shorthand that means
``include all earlier potential outcomes.''  For
example, $C_1(a_1, \ldots)$ means $C_1(a_1, W_1(a_1, U, C_0), U, C_0)$.
By definition of node intervention responses, this counterfactual is equal to
$Y_2(a_1, a_2)$, and our overall joint distribution over the responses is
$p(Y_2(a_1, a_2), A_1, A_2)$.

For arbitrary sets of treatments ${\bf A}$ and outcomes ${\bf Y}$,
and active treatment values ${\bf a}$, we may still represent ETT
as a single mean difference, for example
$\mathbb{E}[f({\bf y})]_{[p({\bf Y}({\bf a}) \mid {\bf a})]} -
\mathbb{E}[f({\bf y})]_{[p({\bf Y}({\bf a}') \mid {\bf a})]}$,
for some function $f({\bf y})$.

Note that though ETT resembles the total effect, it is in fact a more
complex kind of counterfactual.  This is because we are simultaneously
interested in ``outcome responses'' ${\bf Y}$, and ``treatment responses''
${\bf A}$.  Defining these treatment responses may introduce conflicts among
intermediate counterfactual responses, not well represented by node
interventions, which is why we represent ETT as a response
to a path intervention.

The \emph{ETT path intervention}
$\pi^{\bf a}_
{{\mathfrak a}_{{\bm\alpha}_{{\bf A},{\bf Y},{\cal G}}}}$
simply assigns all paths in
${\bm\alpha}_{{\bf A},{\bf Y},{\cal G}}$ to the appropriate
value.  That is, paths from ${\bf A}$ to ${\bf A}$ are assigned
the appropriate natural value, and paths from ${\bf A}$ to ${\bf Y}$ are
assigned the appropriate value in ${\bf a}$.
Given this definition, either the ETT is not identified, or
the joint distribution from which ETT is obtained
corresponds to the joint response of ${\bf Y} \cup {\bf A}$ to the
ETT path intervention.
\begin{lem}
If there exists $A \in {\bf A}$ such that
$A({\mathfrak a}_{{\bm\alpha}_{{\bf A},{\bf Y},{\cal G}}}) \neq A$,
$p({\bf Y}({\bf a}), {\bf A})$ is not identified under the MWM for
${\cal G}$.  If there does not exist such an $A$,
$p({\bf Y}({\bf a}), {\bf A}) = p(\{ {\bf Y} \cup {\bf A} \}
({\mathfrak a}_{{\bm\alpha}_{{\bf A},{\bf Y},{\cal G}}}))$.
\label{lem:ett-translation}
\end{lem}
If $p({\bf Y}({\bf a}), {\bf A})$ is expressible as a response to a path
intervention, it may still not be identifiable under the MWM.

Our subsequent results on identification of path interventions under the
MWM complement identification results in
\cite{identification14jessica,thomas13swig}.  In particular, our results
imply the distribution $p(Y(a,m) \mid A,M)$ is identified under the MWM for
Fig. \ref{fig:triangle} (a), but not under the SWM for
Fig. \ref{fig:triangle} (a).

\subsection{Path-Specific Effects}

Next, we consider the mediated effect of $A_1,A_2$ on $Y_2$ through $C_1,C_2$,
in other words, the effect of exposures on outcome mediated by adherence.
Originally these kinds of effects were considered in \cite{baron86mm} in the
context of linear models, and were generalized to a form not restricted by
particular parametric models in \cite{robins92effects}.  We discuss a simple
version of mediated effects in the graph in Fig. \ref{fig:triangle} (a),
known as natural direct and indirect effects
\cite{robins92effects,pearl01direct} in Section \ref{sec:dir-indir-edge},
where we represented them as responses to edge interventions.

In our case, we are interested in a more complicated effect, but we can
represent it using a similar idea using paths rather than edges -- paths we are
interested in are assigned active treatment values $a_1,a_2$, while paths we
are not interested in are assigned baseline treatment values $a'_1,a'_2$.  The
paths we are interested in are all directed paths with the first edges are one
of $\{ (A_1C_1)_{\to}, (A_1C_2)_{\to}, (A_2C_2)_{\to} \}$, which end in $Y_2$,
and which do not proceed through $A_2$ if started at $A_1$.  The
paths we are not interested in are all other paths which start with $A_2$ or
$A_1$ (and do not proceed through $A_2$) and end in $Y_2$.  
Call this assignment ${\mathfrak a}_1$.  Note that the assignment
${\mathfrak a}_1$ is on the set of paths that is precisely equal to
${\bm\alpha}_{\{ A_1,A_2\},\{Y_2\},{\cal G}}$ for ${\cal G}$ that is the
transitive closure with respect to blue edges of the graph in Fig.
\ref{fig:triangle} (b), and thus is proper.

We apply our definition to obtain a response of $Y_2$ to this intervention.
We must substitute a value for every parent of $Y_2$.  The values for
$A_1,A_2$ will be the baseline $a'_1, a'_2$, while the values for $C_0,U$ will
just be the natural values of those variables.  Complications arise for other
parents, due to the recursive nature of the definition.  We proceed recursively:
\begin{align*}
Y_2({\mathfrak a}_1) &= Y_2(a'_1, a'_2,
	\{ C_2, W_2, Y_1, C_1, W_1, C_0 \}({\mathfrak a}_1), U) \\
C_2({\mathfrak a}_1) &= C_2(a_1, a_2,
	\{ W_2, Y_1, C_1, W_1, C_0 \}({\mathfrak a}_1), U) \\
W_2({\mathfrak a}_1) &= W_2(a'_1, a'_2,
	\{ Y_1, C_1, W_1, C_0 \}({\mathfrak a}_1), U) \\
Y_1({\mathfrak a}_1) &= Y_1(a'_1, \{ C_1, W_1, C_0 \}({\mathfrak a}_1), U) \\
C_1({\mathfrak a}_1) &= C_1(a_1, \{ W_1, C_0 \}({\mathfrak a}_1)) \\
W_1({\mathfrak a}_1) &= W_1(a'_1, C_0({\mathfrak a}_1), U) \\
C_0({\mathfrak a}_1) &= C_0(U) = C_0 \\
\end{align*}

In the matter similar to direct and indirect effects, we can use this response
along with the total effect responses to define
``the effect along paths we want'' as
$\mathbb{E}[Y({\mathfrak a}_1)] - \mathbb{E}[Y(a'_1, a'_2)]$,
and ``the effect along paths we do not want'' as
$\mathbb{E}[Y(a_1, a_2)] - \mathbb{E}[Y({\mathfrak a}_1)]$.
As before, the ACE additively decomposes into these two effect measures.
This definition (without the use of path interventions) appears in
\cite{shpitser13cogsci}.

We may also consider a response of $Y_2$ where the paths we are not interested
in are assigned the natural values, as discussed in Section \ref{sec:natural},
rather than fixed baseline values.  Such an effect is defined similarly.

Consider a set of active treatment values ${\bf a}$ of ${\bf A}$, a set of
fixed baseline treatment values ${\bf a}'$, and a subset ${\bm\beta}$ of
${\bm\alpha}_{{\bf A},{\bf Y},{\cal G}}$ (which contains
``paths of interest'').  Define the \emph{fixed baseline PSE path intervention} 
$\pi^{{\bf a},{\bf a}',{\bm\beta}}_{{\mathfrak a}_{
{\bm\alpha}_{{\bf A},{\bf Y},{\cal G}}}}$
as a path intervention that assigns appropriate active values in ${\bf a}$ to
sources in ${\bm\beta}$ and appropriate baseline values in ${\bf a}'$ to
sources of all paths in
${\bm\alpha}_{{\bf A},{\bf Y},{\cal G}} \setminus {\bm\beta}$.

Similarly, we call an intervention
$\pi^{{\bf a},{\bm\beta}}_{{\mathfrak a}_{
{\bm\alpha}_{{\bf A},{\bf Y},{\cal G}}}}$ that
assigns active values in ${\bf a}$ to sources of paths in
${\bm\beta}$
and appropriate \emph{natural} values to sources of
all paths in ${\bm\alpha}_{{\bf A},{\bf Y},{\cal G}} \setminus {\bm\beta}$
the \emph{average baseline PSE path intervention}.

Path specific effects along all paths in ${\beta}$ (with a fixed baseline)
can then be defined on the mean difference scale as
$\mathbb{E}[{\bf Y}({\mathfrak a}^{{\bf a},{\bf a}',{\bm\beta}}_{
{\bm\alpha}_{{\bf A},{\bf Y},{\cal G}}})]
- \mathbb{E}[{\bf Y}({\bf a}')]$,
and along all paths not in ${\beta}$ as
$\mathbb{E}[{\bf Y}({\bf a})] -
\mathbb{E}[{\bf Y}({\mathfrak a}^{{\bf a},{\bf a}',{\bm\beta}}_{
{\bm\alpha}_{{\bf A},{\bf Y},{\cal G}}})]$.
Average baseline path specific effects on the difference scale are defined
similarly.

\subsection{Effects of Treatment on the Indirectly Treated}
\label{sec:ett-innoc}

In this section we show that the language of path interventions is general
enough to incorporate novel targets not currently considered in the literature.
Our results immediately settle identification questions for any such target.

We consider a seemingly innocuous ETT with two treatments
that in fact can only be represented by a path intervention, not an edge
intervention, and variations of this target that are identified under the
SWM and the MWM.
Assume Fig. \ref{fig:triangle} (a) represents a simple two time point partially
randomized observational study, where $W$ and $M$ are treatments at the first
and second time points, respectively, $A$ is an intermediate health measure,
and $Y$ is the outcome.  We make very strong assumptions about this study.
In particular, $W$ is randomized, while $M$ is only assigned based on $A,W$.
Finally, no unobserved confounding exists anywhere, including between $W,M$ and
$Y$.  We are interested in the effect of treatments $W,M$ on the treated in
this study.  To obtain this contrast, we need to identify $p(Y(m,w) \mid W,M)$
which is identified if and only if $p(Y(m,w),W,M)$ is.
It is not difficult to show that
\begin{align*}
p(Y(m,w),W,M) &=
p(\{ Y,M,W \}((wAY)_{\to},(mY)_{\to}))\\
&= p(Y(m,A(w)),M(A(W),W),W).
\end{align*}

As we will show in the next section, there is no way to express this
response as a response to an edge intervention, and it is not
identified under the MWM.  This is the case despite the fact that
there is no unobserved confounding in this study.  The difficulty is that
the response is defined in terms of $A(w)$ and $A$ \emph{jointly}, and the
distribution $p(A(w),A)$ is not identified under the MWM
without more assumptions.

To obtain a target that is identified under the SWM in this case
we may consider the response $Y(w,m)$ on the treated to the natural value $W$,
and the value of $M$ occurring under the intervention setting $W$ to $w$.
This results in $p(Y(m,w),W,M(w)) = p(Y(m,A(M(w))),M(A(w)),W)$ which is then
identified under the SWM.  To obtain identification we gave up
on conditioning on the natural value of the second treatment $M$.  This may
not be ``in the spirit'' of the ETT target.

One compromise is to assume a stronger model, the MWM, and allow the
response $M$ to be ``as natural as possible'' while still retaining
identification.  This would correspond to defining a contrast in terms of
$p(\{ Y, W, M \}((mY)_{\to}, (wAY)_{\to}, (wAM)_{\to}))$,
which in turn is equivalent to
$p(\{ Y, W, M \}((mY)_{\to}, (wA)_{\to}))$.
A conditional distribution\\
$p(Y((mY)_{\to}, (wA)_{\to}), M((wA)_{\to}, (w'M)_{\to}) \mid W = w')$
represents the response $Y(w,m)$ among those
individuals whose treatment value for $W$ is $w'$ (untreated), and
whose treatment value for $M$ is whatever value $M$ would have attained had
$W$ assumed the active value $w$ with respect to the path $(WAM)_{\to}$, and
untreated value $w'$ with respect to the path $(WM)_{\to}$.

We can define a contrast based on this quantity, using a summary function
$f(Y,M)$, equal to
{\small
\[
\mathbb{E}[f(Y(m,w), M(A(w),w')) \mid W=w'] -
\mathbb{E}[f(Y(m',w'), M(A(w'),w')) \mid W=w'],
\]
}
which we call
``the effect of treatment on the indirectly treated (ETIT).''  The name is due
to the fact that we consider people whose baseline treatment $W$ is untreated,
and whose followup treatment $M$ is set to a value that is a kind of response
to the indirect effect of the first treatment.
Such a quantity would be difficult to conceive of without a direct
representation of effects along pathways, something path interventions provide.
Our results also directly imply that this quantity is identified under the MWM,
but not SWM.

\section{Identification of Edge and Path Interventions}
\label{sec:id}

Having established a correspondence between responses to path interventions
and a variety of targets of interest in causal inference, we now
consider what assumptions are necessary to express path interventions as
edge interventions, edge interventions as node interventions, and edge and
node interventions as functions of the observed data.

As we showed in section \ref{sec:natural}, we can restrict our attention to
path interventions that only assign paths to constant values, since paths
that are assigned natural values either can be dropped from the intervention
without affecting the response, or the overall response is not identified.

\subsection{Node and Edge Interventions as Path Interventions}

If node interventions are a special case of edge interventions, which are in
turn a special case of path interventions, we ought to be able to give a path
intervention the response of which is equal to the response to an arbitrary node
or edge intervention.  For any such response there may be
multiple path interventions the responses to which are identical.
We give one such path intervention here.

\begin{lem}
Let ${\bf A},{\bf Y}$ be disjoint vertex sets in a DAG ${\cal G}$, and
${\bf a}$ a value assignment to ${\bf A}$.
Let $\pi^{\bf a}_{{\mathfrak a}_{{\bm\alpha}_{{\bf A},{\bf Y},{\cal G}}}}$
assign each $\alpha \in {\bm\alpha}_{{\bf A},{\bf Y},{\cal G}}$ to
${\bf a}_{\sop_{\cal G}(\alpha)}$.  Then
$p({\bf Y}({\mathfrak a}_{{\bm\alpha}_{{\bf A},{\bf Y},{\cal G}}})) =
p({\bf Y}({\bf a}))$.
\label{lem:node-as-path}
\end{lem}

\begin{lem}
Let ${\bf Y}$ be a vertex set in a DAG ${\cal G}$, and $\bm\alpha$ a set of
edges, with ${\mathfrak a}_{\bm\alpha}$ an assignment to ${\bm\alpha}$.
Let ${\bf A} = \sop_{\cal G}({\bm\alpha})$, and
${\bm\alpha}_{{\bf Y},{\cal G}}$ be a subset of
${\bm\alpha}_{{\bf A},{\bf Y},{\cal G}}$ consisting of paths with an edge prefix
in ${\bm\alpha}$.
Let
$\pi^{\bm\alpha}_{{\mathfrak a}_{{\bm\alpha}_{{\bf Y},{\cal G}}}}$
assign each $\alpha \in {\bm\alpha}_{{\bf Y},{\cal G}}$
to the value assigned to the edge prefix of $\alpha$
by ${\mathfrak a}_{\bm\alpha}$.  Then
$p({\bf Y}({\mathfrak a}_{{\bm\alpha}_{{\bf Y},{\cal G}}})) =
p({\bf Y}({\mathfrak a}_{\bm\alpha}))$.
\label{lem:edge-as-path}
\end{lem}

\subsection{Identification of Edge Interventions}
\label{sec:edge-interventions}

The difficulty with edge interventions is that a single response to an edge
intervention may involve other responses with conflicting treatment
assignments.  It is this feature of edge interventions which in general
prevents their identification under the SWM, and which requires the
stronger assumptions of the MWM.  If such a conflicting assignment is
absent, the edge intervention can be rephrased as a node intervention.  We show
this absence of conflict is characterized by a property we call node
consistency.

A set of edges ${\bm\alpha}$ live for ${\bf Y}$ is called \emph{consistent} for
${\bf Y}$ if for every node $A$, the set of prefix edges of the path set
$\left\{ \alpha \in \rel_{\cal G}({\bf Y}\mid{\bm\alpha}) \middle|
\sop_{\cal G}(\alpha) = A \right\}$
is either disjoint from $\bm\alpha$ or contained in $\bm\alpha$.

For a set of edges ${\bm\alpha}$ live and consistent for ${\bf Y}$, we call
an edge intervention $\eta_{{\mathfrak a}_{\bm\alpha}}$
\emph{node consistent} for ${\bf Y}$ if
for every node $A$, all edges in ${\bm\alpha}$ with $A$ as the source node
are assigned the same value (say $a$).  Any edge intervention that is not
node consistent we call node inconsistent, including any edge intervention
on a set of edges not consistent for an outcome set of interest.

The edge set $\{ (AY)_{\to} \}$ in Fig. \ref{fig:triangle} (a) is
live but not consistent for $\{ Y \}$, thus any edge intervention on this
set (that sets to constant values) is inconsistent for $\{ Y \}$.
An edge intervention corresponding
to $Y((aY)_{\to}, (aM)_{\to})$ is node consistent for $Y$, while
an edge intervention corresponding to $Y((aY)_{\to}, (a'M)_{\to})$ is
consistent, but not node consistent for $Y$.

For an edge intervention
$\eta_{{\mathfrak a}_{\bm\alpha}}$ node consistent for ${\bf Y}$,
define the following set of value assignments to
${\bf A} = \sop_{\cal G}({\bm\alpha})$,
${\bf a}_{\bm\alpha} \equiv
	\{ a \mid \eta \text{ assigns $a$ to } (AB)_{\to} \in {\bm\alpha} \}$.
Let $\nu_{{\bf a}_{\bm\alpha}}$ be the \emph{induced node intervention} for
$\eta_{{\mathfrak a}_{\bm\alpha}}$.

\begin{lem}
Given a DAG ${\cal G}$ with vertices ${\bf V}$, and an edge
intervention $\eta_{{\mathfrak a}_{\bm\alpha}}$ node consistent for
${\bf Y} \subseteq {\bf V}$,
$p({\bf Y}({\mathfrak a}_{\bm\alpha})) = p({\bf Y}({\bf a}_{\bm\alpha}))$.
\label{lem:node-cons}
\end{lem}
\begin{proof}
This follows by lemmas \ref{lem:node-as-path} and \ref{lem:edge-as-path}.
\end{proof}

\begin{cor}
If $\eta_{{\mathfrak a}_{\bm\alpha}}$ is node consistent for ${\bf Y}$,
then $p({\bf Y}({\mathfrak a}_{\bm\alpha}))$ is identified as a functional of
$p({\bf V})$ under the SWM via the extended g-formula
(\ref{eqn:g-formula}) for the response to the corresponding induced node
intervention.
\label{cor:g-formula-edge}
\end{cor}

We next show that if an edge intervention is not node consistent, then
responses to this intervention are not identifiable from $p({\bf V})$ under
the SWM.  By this we mean that the definition of identifiability given in
Section \ref{sec:node-id} fails, and more specifically that we can find
two elements of a causal model, in the sense of section \ref{sec:models},
that agree on $p({\bf V})$ but disagree on
the distribution of the response of interest.
We start with a simple example of a non-identified parameter in the
SWM.

\begin{figure}
\begin{center}
  \begin{tikzpicture}[>=stealth, node distance=1.0cm]
    \tikzstyle{square} = [draw, very thick, rectangle, minimum size=5mm]
    \tikzstyle{format} = [draw, very thick, circle, minimum size=5.0mm,
	inner sep=0pt]
  \begin{scope}
    \path[->, very thick]
		node[format] (a) {$A$}
		node[format, below of=a] (b) {$B$}

		(a) edge[blue] (b)
		;
		\node[xshift=0.0cm, yshift=0.7cm] {(a)}
		;
  \end{scope}
  \begin{scope}[xshift=3.0cm]
    \path[->, very thick]
		node[format] (a) {$A$}
		node[format, below left of=a] (b) {$B$}
		node[format, below right of=a] (c) {$C$}

		(a) edge[blue] (b)
		(a) edge[blue] (c)
		;
		\node[xshift=0.0cm, yshift=0.7cm] {(b)}
		;
  \end{scope}
  \begin{scope}[xshift=7.0cm]
    \path[->, very thick]
		node[format] (a) {$A$}
		node[format, below left of=a] (b) {$B$}
		node[format, below right of=a] (c) {$C$}
		node[format, above left of=b] (ub) {$U_B$}
		node[format, above right of=c] (uc) {$U_C$}

		(a) edge[blue] (b)
		(a) edge[blue] (c)
		(ub) edge[blue] (b)
		(uc) edge[blue] (c)
		;
		\node[xshift=0.0cm, yshift=0.7cm] {(c)}
		;
  \end{scope}
  \begin{scope}[xshift=11.0cm]
    \path[->, very thick]
		node[format] (a) {$A$}
		node[format, below left of=a] (b) {$B$}
		node[format, below right of=a] (c) {$C$}
		node[format, above left of=b] (ub) {$U_B$}
		node[format, above right of=c] (uc) {$U_C$}

		(a) edge[blue] (b)
		(a) edge[blue] (c)
		(ub) edge[blue] (b)
		(ub) edge[blue] (c)
		(uc) edge[blue] (b)
		(uc) edge[blue] (c)
		;
		\node[xshift=0.0cm, yshift=0.7cm] {(d)}
		;
  \end{scope}

  \end{tikzpicture}
\end{center}
\caption{}
\label{fig:non-id}
\end{figure}

\begin{lem}
Responses
$p(\{ B, C \}((aB)_{\to},(a'C)_{\to})$,
$p(\{ B, C \}((aB)_{\to})$,\\
and $p(\{ B, C \}((aC)_{\to})$, are not identifiable from
$p(A,B,C)$ under the SWM for Fig. \ref{fig:non-id} (b).
\label{lem:non-id-edge}
\end{lem}

The proofs of this result, which appears in the appendix, exhibits two causal
structures $c_1(\{A,B,C\},{\cal G})$, and $c_2(\{A,B,C\},{\cal G})$ that
agree on\\
$p(A,B,C)$, but disagree on the above responses to (node inconsistent)
edge interventions.  These two structures corresponding to graphs in
Fig. \ref{fig:non-id} (c), (d).  In particular, $c_2$ is constructed in such a
way that the confounding of $B$ and $C$ introduced by $U_B$ and $U_C$ is masked
under any single node intervention, but manifests if we consider responses to
multiple interventions simultaneously.  This is similar in spirit to an example
in \cite{robins10alternative}.
We can extend this simple example to a general result, due to the following
lemma (stated in a more general form in terms of path rather than edge
interventions).

\begin{lem}
Let ${\cal G}$ be a DAG, ${\bf Y}$,${\bf A}$ disjoint sets of vertices in
${\cal G}$, ${\bm\alpha}$ a set of live directed paths proper for ${\bf Y}$.
Let ${\cal G}^*$ be any edge supergraph of ${\cal G}$, ${\bf Y}^*$ any
superset of ${\bf Y}$ in ${\cal G}^*$, ${\bm\alpha}^*$ a superset of
${\bm\alpha}$ in ${\cal G}^*$ live and proper for ${\bf Y}^*$, such that
every path in ${\bm\alpha}^*\setminus{\bm\alpha}$ does not exist in ${\cal G}$.
Finally, let $\pi_{{\mathfrak a}_{{\bm\alpha}^*}}$ be a path intervention.
If $p({\bf Y}({\mathfrak a}_{{\bm\alpha}}),{\bf A})$ is not identified under
the MWM (SWM) for ${\cal G}$, then
$p({\bf Y}^*({\mathfrak a}_{{\bm\alpha}^*}),{\bf A})$ and
$p({\bf Y}({\mathfrak a}_{{\bm\alpha}^*}),
{\bf A} \cup {\bf Y}^*\setminus{\bf Y})$
are not identified under the MWM (SWM) for ${\cal G}^*$.
\label{lem:small-to-large}
\end{lem}

\begin{thm}
Consider a DAG ${\cal G}$ with vertices ${\bf V}$, and a set of edges
${\bm\alpha}$ live for ${\bf Y}$.  Then $p({\bf Y}({\mathfrak a}_{\bm\alpha}))$
is identifiable from $p({\bf V})$ under the SWM for ${\cal G}$ if
and only if $\eta_{{\mathfrak a}_{\bm\alpha}}$ is node consistent.
Moreover, if $p({\bf Y}({\mathfrak a}_{\bm\alpha}))$ is identifiable, it is
given by the extended g-formula (\ref{eqn:g-formula})
for $p({\bf Y}({\bf a}_{\bm\alpha}))$, the response to the
induced node intervention.
\label{thm:non-id-edge}
\end{thm}

What we have shown is that node consistent edge interventions are
identifiable under the SWM, but an edge
intervention that is node inconsistent is not, as long as
this inconsistency is ``causally relevant'' for some response, in the sense
of there existing causal pathways from the inconsistent edges to some
responses that are not interrupted by other parts of the
edge intervention.  However, if we are willing to adopt stronger independence
assumptions of the MWM, we obtain identification of any edge intervention
via a modification of the g-formula, as the following result shows.

\begin{lem}[edge g-formula]
For a DAG ${\cal G}$ with vertices ${\bf V}$, and an edge intervention
$\eta_{{\mathfrak a}_{\bm\alpha}}$ on an edge set ${\bm\alpha}$,
we have, under the MWM for ${\cal G}$,
\begin{align}
p({\bf V}({\mathfrak a}_{\bm\alpha}) = {\bf v}) = \prod_{V \in {\bf V}}
	p(V = {\bf v}_V \mid {\bf v}_{\pa_{\cal G}^{\overline{\bm\alpha}}(V)},
	{\mathfrak a}_{\{ (WV)_{\to} \in {\bm\alpha} \}}),
\label{eqn:g-formula-edge}
\end{align}
where $\pa_{\cal G}^{\overline{\bm\alpha}}(V) \equiv
\{ A \in \pa_{\cal G}(V) \mid (AV)_{\to} \not\in {\bm\alpha} \}$.
\label{lem:g-formula-edge}
\end{lem}

For example, in the graph in Fig. \ref{fig:triangle} (a),
we can express the distribution of the response of
$Y((a'M)_{\to}, (aY)_{\to})$
using (\ref{eqn:g-formula-edge}) as follows:
\begin{align*}
p(Y(a,M(a')) = y) &=
	\sum_{w,a'',m} p(y \mid m,a) p(m \mid a',w) p(a'' \mid w) p(w)\\
	&= \sum_{m,w} p(y \mid m,a) p(m \mid a',w) p(w)
\end{align*}
If we are interested in a mean difference parameter, for example
$\mathbb{E}[Y(a,M(a'))] - \mathbb{E}[Y(a')]$, and assume there are no baseline
factors $W$, the above reduces to
\[
\sum_{m} \left\{ \mathbb{E}[Y \mid m,a] -
\mathbb{E}[Y \mid m,a'] \right\} p(m \mid a')
\]
which recovers the well known \emph{mediation formula} \cite{pearl11cmf}.

The independence assumptions which were necessary to derive this functional,
namely $(Y(m,a) \ci M(a') \ci A)$, are implied by the MWM for the graph
in Fig. \ref{fig:triangle} (a).  It is possible to consider
such assumptions independently of a graph.  However the advantage of graphs
is their ability to encode assumptions of this type \emph{systematically},
which allowed us to derive such functionals for a wide variety of problems,
and moreover, to give simple visual characterizations of when
such derivations are possible.

\subsection{Identification of Path Interventions}

As we saw in the previous section, identification of responses to edge
interventions under the SWM requires node consistency, while any
joint response to any edge intervention is identified under the MWM.  In
this section we show that path interventions are identified under the MWM
as long as \emph{edge consistency} holds, that is as long as a path
intervention can be expressed as an edge intervention.  Lack of edge consistency
will result in non-identification under the MWM.  The presence of a
``recanting witness'' in a path-specific effect \cite{chen05ijcai} can be
viewed as a special case of the lack of edge consistency.


A set of directed paths ${\bm\alpha}$ live for ${\bf Y}$ is called
\emph{consistent} for ${\bf Y}$ if for every edge $(AB)_{\to}$ that is
an edge prefix of $\alpha \in {\bm\alpha}$, if $(AB)_{\to}$
is in $\beta \in \rel_{\cal G}({\bf Y}\mid{\bm\alpha})$, then
$(AB)_{\to}$ is an edge prefix of a prefix subpath of $\beta$
in ${\bm\alpha}$.

For a proper set of directed paths ${\bm\alpha}$ live and consistent for
${\bf Y}$, we call a path intervention
$\pi_{{\mathfrak a}_{\bm\alpha}}$ \emph{edge consistent} for
${\bf Y}$ if for every edge $(AB)_{\to}$, all paths in ${\bm\alpha}$ with
$(AB)_{\to}$ as a prefix are assigned the same value (say $a$).
Any path intervention that is not edge consistent we call edge inconsistent,
including any path intervention on a set of paths not consistent for an outcome
set of interest.

The path set $\{ (WAMY)_{\to} \}$ in Fig. \ref{fig:triangle} (a) is live but not
consistent for $\{ Y \}$, thus any path intervention on this set is
inconsistent for $\{ Y \}$.  A path intervention corresponding to
$Y((wAMY)_{\to}, (wAY)_{\to})$ is edge consistent for $Y$, while a
path intervention corresponding to
$Y((wAMY)_{\to}, (w'AY)_{\to})$ is consistent for $Y$, but not edge
consistent for $Y$.

For a path intervention $\pi_{{\mathfrak a}_{\bm\alpha}}$ edge consistent
for ${\bf Y}$, define the set of edges
${\bm\alpha}_1 \equiv \{ (AB)_{\to} \mid
	(AB)_{\to} \text{ is a prefix for } \alpha \in {\bm\alpha} \}$.
Let $\eta_{{\mathfrak a}_{{\bm\alpha}_1}}$
be the induced edge intervention for $\pi_{{\mathfrak a}_{\bm\alpha}}$,
where $\eta$ assigns $(AB)_{\to} \in {\bm\alpha}_1$ to the
value assigned by $\pi$ to all $\alpha \in {\bm\alpha}$ which have $(AB)_{\to}$
as an edge prefix.

\begin{lem}
Given a DAG ${\cal G}$ with vertices ${\bf V}$, and a path intervention
$\pi_{{\mathfrak a}_{\bm\alpha}}$ edge consistent for ${\bf Y} \subseteq
{\bf V}$, $p({\bf Y}({\mathfrak a}_{{\bm\alpha}})) =
p({\bf Y}({\mathfrak a}_{{\bm\alpha}_1}))$.
\label{lem:edge-cons}
\end{lem}

\begin{cor}
If $\pi_{{\mathfrak a}_{\bm\alpha}}$ is edge consistent for ${\bf Y}$,
then the distribution $p({\bf Y}({\mathfrak a}_{{\bm\alpha}}))$ is identified
as a functional of $p({\bf V})$ under the MWM model via the edge
g-formula for the response to the corresponding induced edge intervention.
\label{cor:g-formula-path}
\end{cor}




We will show that responses to edge inconsistent path interventions are not
identifiable under the MWM using the same strategy as we used for
node inconsistent edge interventions.  First, we reproduce a result
stating that a joint response to a conflicting exposure is not identifiable.
Then we extend this result to the general case we need.


\begin{lem}
The distributions $p(B(a), B(a'))$ and $p(B(a), B)$ are not identifiable
from $p(A, B)$
under the MWM for the DAG in
Fig. \ref{fig:non-id} (a).
\label{lem:chen}
\end{lem}

\begin{thm}
Consider a DAG ${\cal G}$ with vertices ${\bf V}$, and a proper set of paths
${\bm\alpha}$ live for ${\bf Y}$.  Then
$p({\bf Y}({\mathfrak a}_{\bm\alpha}))$ is identifiable from
$p({\bf V})$ under the MWM for ${\cal G}$ if and only if
$\pi_{{\mathfrak a}_{\bm\alpha}}$ is edge consistent.  Moreover, if
$p({\bf Y}({\mathfrak a}_{\bm\alpha}))$ is identifiable, it is given by the
edge g-formula for $p({\bf Y}({\mathfrak a}_{{\bm\alpha}_1}))$,
the response to the
induced edge intervention.
\label{lem:non-id-path}
\end{thm}

\subsection{A Model Where Responses to Path Interventions Are Identified}
\label{sec:linear-sem}
Though we have shown that responses to path interventions that cannot be
expressed as responses to edge interventions are not in general identified
under the MWM, there exist submodels of the MWM where all responses to
path interventions are identified.  In particular, consider the \emph{linear
structural equation model (SEM)}, which is an MWM where the mapping from
${\bf v}_{\pa_{\cal G}(V)} \in {\mathfrak X}_{\pa_{\cal G}(V)}$ to
$V({\bf v}_{\pa_{\cal G}(V)})$ is a linear function of
${\bf v}_{\pa_{\cal G}(V)}$ and an error term $\epsilon_V$, where such
error terms are normally distributed and mutually independent.

\begin{thm}
Let $\pi_{{\mathfrak a}_{\bm\alpha}}$ be a path intervention.  Then
$p({\bf Y}({\mathfrak a}_{\bm\alpha}))$ is identified under the linear SEM.
\end{thm}

This follows as a corollary of results in \cite{balke94countereval}.  The
reason even edge-inconsistent path interventions are identified is that
linearity, normality and independence are such strong assumptions that we can
directly evaluate even counterfactuals of the form $p(W(a),W(a'))$ using the
algorithm in \cite{balke94countereval}.
A fruitful open question if whether there are other interesting (for instance
maximal) submodels of the MWM where all responses to path interventions
are identified.

\subsection{Targets Not Representable as Path Interventions}

We have shown that a wide class of targets of interest in causal inference
can be expressed as responses to path interventions.  Nevertheless, there
exist targets of interest which are known not to be representable in this way,
such as principal stratification effects.  For instance,
the principal stratum direct effect (PSDE) \cite{rubin04direct,rubin05causal}
is defined to be
a treatment contrast only among those individuals for whom the mediator
assumes a particular value for both active and baseline treatment levels.
In Fig. \ref{fig:triangle} (a), the PSDE is a contrast of the form
\[
\mathbb{E}[Y(a,m) \mid M(a) = M(a') = m] -
\mathbb{E}[Y(a',m) \mid M(a) = M(a') = m].
\]
Under the MWM, we obtain independences $Y(a,m) \ci \{ M(a), M(a') \}$,
and $Y(a',m) \ci \{ M(a), M(a') \}$, which implies the PSDE is equal to the
controlled direct effect contrast under the MWM:
$\mathbb{E}[Y(a,m)] - \mathbb{E}[Y(a',m)]$.
Under the SWM, the PSDE contrast is not identified without more assumptions.
In either case, it is not possible to express the condition defining the
principal strata, namely $M(a) = M(a') = m$ as a response
to a path intervention, since this will entail assigning conflicting
values to a directed edge from $A$ to $M$.  This is perhaps not surprising,
since responses to path interventions are meant to encode effects \emph{along
particular causal pathways} which is not something principal strata effects
encode.  Note that despite this, the MWM allows us to rephrase the PSDE as a
node intervention.

\section{The Edge G-Formula and Single World Intervention Graphs}
\label{sec:swigs}

A connection between the SWM, node interventions, the extended g-formula,
and a type of graph with split nodes called the Single World Intervention Graph
(SWIG) was given in \cite{thomas13swig}.

If a set of responses ${\bf V}$ to a node intervention $\nu_{\bf a}$ includes
all variables (including ${\bf A}$), then, under the SWM, the response is linked
to the observed distribution via (\ref{eqn:g-formula}), and can be viewed
as a kind of Markov factorization \cite{pearl88probabilistic} of the
joint response ${\bf V}({\bf a})$, where terms
$p(V \mid \pa_{\cal G}(V))$ with $\pa_{\cal G}(V) \cap {\bf A} \neq \emptyset$
are replaced with $p(V \mid \pa_{\cal G}(V) \setminus {\bf A},
{\bf a}_{\pa_{\cal G}(V) \cap {\bf A}})$.  SWIGs are a graphical representation
of this factorization, in the sense that independences in $p({\bf V}({\bf a}))$
can be read off from the corresponding SWIG.
Since $A$ occurs both as a treatment and a response, SWIGs split the vertex $A$
into a random and fixed versions (we draw fixed vertices as squares).

For example, the SWIG in Fig. \ref{fig:swig} (a) represents
$p(\{ Y,M,W,A\}(a))$ in the SWM corresponding to Fig. \ref{fig:triangle} (a).
We can check independences of counterfactuals in the joint $p(\{Y,M,W,A\}(a))$,
via a simple modification of the d-separation criterion
\cite{pearl88probabilistic}.  For instance, $Y(a) \ci A \mid W$,
since all d-connected paths from $Y$ to $A$ are blocked by $W$.

Similarly, if a set of responses ${\bf V}$ to an edge intervention
$\eta_{\mathfrak a}$ includes all variables (including ${\bf A}$), then, under
the MWM, the response is linked to the observed distribution via
(\ref{eqn:g-formula}), and can be viewed
as a kind of Markov factorization \cite{pearl88probabilistic} of the joint
response ${\bf V}({\mathfrak a})$, where terms
$p(V \mid \pa_{\cal G}(V))$ with
$\pa_{\cal G}(V) \cap \sop_{\cal G}({\bm\alpha}) \neq \emptyset$
are replaced with $p(V \mid \pa^{\overline{\bm\alpha}}_{\cal G}(V),
	{\mathfrak a}_{(WV)_{\to} \in {\bm\alpha}})$.
It is possible to generalize SWIGs to give a graphical representation of this
factorization.  Instead of splitting the vertices into the fixed and
random versions, we instead shatter every intervened-on vertex into a set
corresponding to distinct values (including the natural value) that vertex
assumes when defining the response.  For example, the graph in
Fig. \ref{fig:swig} (b) represents
$p(\{ Y,M,W,A\}((wM)_{\to},(w'A)_{\to},(aY)_{\to}))$ in the MWM
corresponding to Fig. \ref{fig:triangle} (a).  We can check independences of
counterfactuals in this joint via a simple modification of d-separation:
$Y((aY)_{\to},(wM)_{\to},(w'A)_{\to}) \ci A((w'A)) \mid
M((w'A)_{\to},(wM)_{\to})$ since all d-connected paths from $Y$ to $A$ are
blocked by $M$.  Note that we shatter $W$ in Fig. \ref{fig:triangle} (a) into
three vertices, and $A$ into two, where the random vertex has an outgoing
arrow to $M$.  This is because there are two treatment values for $W$, and
$W$ is also a response, while $A$ is a response for the purposes of the
$(AM)_{\to}$ edge and a treatment for the purposes of the $(AY)_{\to}$ edge.
That responses to edge interventions factorize according to these kinds of
``shattered graphs'' under the MWM (but not SWM)
follows as a straightforward generalization of the proof
of proposition 11 in \cite{thomas13swig}.  In fact, these shattered graphs
can be viewed as SWIGs defined on an augmented graph where a treatment
vertex is split into copies, corresponding to (individually intervenable)
components of the treatment associated with direct and indirect effects.
For examples of such graphs, and associated discussion, see
\cite{robins10alternative}, Section 6, and Fig. 6.

Thus, the edge g-formula can be viewed as the MWM analogue of the extended
g-formula, and it is possible to construct graphs that stand in the same
relation to edge interventions, the edge g-formula, and the MWM as SWIGs do to
node interventions, the extended g-formula, and the SWM.
In the interests of space, we do not derive this formally, nor
pursue this connection further here.

\begin{figure}
\begin{center}
\begin{tikzpicture}[>=stealth, node distance=0.9cm]
    \tikzstyle{format} = [draw, very thick, circle, minimum size=5.0mm,
	inner sep=0pt]
    \tikzstyle{square} = [draw, very thick, rectangle, minimum size=5mm]

	\begin{scope}
		\path[->, very thick]
			node[format] (w) {$W$}
			node[format, right of=w] (a) {$A$}
			node[square, right of=a] (ap) {$a$}
			node[format, above right of=ap] (m) {$M$}
			node[format, below right of=m] (y) {$Y$}

			(w) edge[blue] (a)
			(ap) edge[blue] (y)
			(ap) edge[blue] (m)
			(m) edge[blue] (y)
			(w) edge[blue, bend left] (m)

			node[below of=a, yshift=0.2cm, xshift=0.3cm] (l) {$(a)$}
		;
	\end{scope}
	\begin{scope}[xshift=5.0cm]
		\path[->, very thick]
			node[format] (w) {$W$}
			node[square, above right of=w] (w1) {$w$}
			node[square, right of=w] (w2) {$w'$}

			node[format, right of=w2] (a) {$A$}
			node[square, right of=a] (ap) {$a$}
			node[format, above right of=ap] (m) {$M$}
			node[format, below right of=m] (y) {$Y$}

			(w1) edge[blue] (m)
			(w2) edge[blue] (a)

			(ap) edge[blue] (y)
			(a) edge[blue, bend left=10] (m)
			(m) edge[blue] (y)

			node[below of=a, yshift=0.2cm, xshift=0.3cm] (l) {$(b)$}
		;
	\end{scope}
\end{tikzpicture}
\end{center}
\caption{(a) A SWIG for $\{ Y, M, W, A \}(a)$.
(b) An edge intervention version for $\{ Y, M, W, A\}((wM)_{\to},(w'A)_{\to},
(aY)_{\to})$.}
\label{fig:swig}
\end{figure}
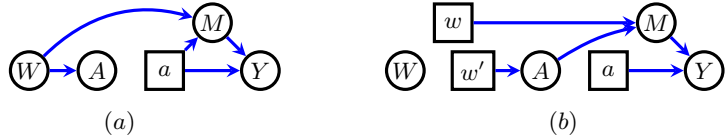

\section{The Edge G-Formula and Causal Effects in Hidden Variable
	DAGs}
\label{sec:tian}

If some variables in a causal model of a DAG are unobserved, not every
response to a node intervention is identified, since (\ref{eqn:g-formula})
cannot always be directly applied.  A complete algorithm for identifying 
${\bf Y}({\bf a})$ where ${\bf A}$ and ${\bf Y}$ are
disjoint in this setting was given in \cite{tian02on,shpitser06id}.
This algorithm, called {\bf ID}, takes as inputs a graph ${\cal G}$
representing a causal model, an observed distribution $p({\bf V})$ for this
model, and disjoint sets ${\bf A}$ and ${\bf Y}$ representing treatments and
outcomes for a causal effect we are interested in.  The algorithm either outputs
a functional of $p({\bf V})$ which is equal to $p({\bf Y}({\bf a}))$ under the
given model, or ``{\bf Not identified.}''
In this section we show that certain outputs of this algorithm
correspond to marginals of the edge g-formula (\ref{eqn:g-formula-edge}).

For example, it can be shown that $p(Y(a))$ is identified via
the \emph{front-door functional}
$\sum_{m,a'} p(Y \mid a',m) p(m \mid a) p(a')$ under the SWM shown in
Fig. \ref{fig:triangle2} (a), where $H$ is not observable.  If we replace
$H$ and its outgoing arrows by an arrow from $A$ to $Y$, we obtain the DAG
in Fig. \ref{fig:triangle2} (c).  A straightforward consequence of
(\ref{eqn:g-formula-edge}) is that $p(Y((aM)_{\to}))$ is identified via the
same functional under the MWM for Fig. \ref{fig:triangle2} (c).  In this
section, we give a general condition for case when this correspondence of
functionals occurs.

Although it is possible to define {\bf ID} on hidden variable DAGs directly,
for convenience it has been defined on acyclic directed mixed graphs (ADMGs).
An ADMG is a mixed graph
with directed $(\to)$ and bidirected edges $(\leftrightarrow)$, with
no directed cycles.  ADMGs represent classes of hidden
variable DAGs via a latent projection operation onto a graph defined only
on observed variables \cite{verma90equiv}.
For example, this operation applied to Fig. \ref{fig:triangle2} (a)
results in an ADMG shown in Fig. \ref{fig:triangle2} (d).
Connected components in a graph obtained from an ADMG ${\cal G}$ by dropping
all directed edges are called \emph{districts} of ${\cal G}$.  For example,
the sets $\{ A, Y \}$ and $\{ M \}$ are districts of the graph in Fig.
\ref{fig:triangle2} (d).
The set of districts of ${\cal G}$ is denoted by ${\cal D}({\cal G})$.
If a set ${\bf S}$ is in a district of ${\cal G}$, we denote that
district by $\dis_{\cal G}({\bf S})$.

For an ADMG ${\cal G}$ with vertices ${\bf V}$, and ${\bf A} \subseteq {\bf V}$,
let ${\cal G}_{\bf A}$ be a subgraph consisting only of vertices in ${\bf A}$
and edges between them.
Let $\an_{\cal G}(V) = \{ A \mid A \to \ldots \to V \text{ is in } {\cal G} \}$.
A total order $\prec_{\cal G}$
on ${\bf V}$ in ${\cal G}$ is \emph{topological}
if whenever $V_1 \prec_{\cal G} V_2$, $V_2 \not\in \an_{\cal G}(V_1)$.
For a total order $\prec$ on ${\bf V}$, for any $V \in {\bf V}$,
let $\pre_{\prec}(V) \equiv \{ W \in {\bf V} \setminus \{ V \} \mid
W \prec V \}$.

In the remainder of this section we will restrict
attention to inputs of ${\bf ID}$ such that
${\bf V} \subseteq \an_{\cal G}({\bf Y})$, and
${\bf V} \setminus \an_{{\cal G}_{{\bf V} \setminus {\bf A}}}({\bf Y})
\subseteq {\bf A}$, and to the
following subset of possible outputs of ${\bf ID}$.
This subset of outputs is particularly nice since it only involves conditional
distributions derived from $p({\bf V})$.

\begin{dfn}
Given $p({\bf V})$, for any total order $\prec$ on ${\bf V}$,
and ${\bf v} \in {\mathfrak X}_{\bf V}$,
a functional of $p({\bf V})$ of the form
$\sum_{{\bf V}' \setminus {\bf Y}} \prod_{V \in {\bf V}'}
	p(V \mid {\bf S}_V, {\bf v}_{\pre_{\prec}(V) \setminus {\bf S}_V})$,
where ${\bf Y} \subseteq {\bf V}' \subseteq {\bf V}$,
and ${\bf S}_V \subseteq \pre_{\prec}(V) \cap {\bf V}'$
is called a \emph{g-functional}.
\end{dfn}

The output of the g-computation algorithm \cite{robins86new}, mentioned in
Section \ref{sec:total-as-node}, is always a g-functional,
but some outputs of ${\bf ID}$ are g-functionals that
cannot arise from g-computation.  For instance, the front-door functional
is a g-functional, but g-computation cannot be used to identify
treatment effects with unobserved common causes as is the case in Fig.
\ref{fig:triangle2} (a).

We give a sufficient condition on the input ADMG ${\cal G}$ to {\bf ID} such
that the output is a g-functional, and then show that it is possible to
construct a DAG ${\cal G}^{\dag}$ from ${\cal G}$ where a certain response to
an edge intervention is identified via the same g-functional via
(\ref{eqn:g-formula-edge}).

For a particular treatment set ${\bf A}$ in ${\cal G}$,
let ${\bf D}_{{\bf S},{\bf A},{\cal G}} =
\dis_{{\cal G}_{\an_{\cal G}({\bf S})}}({\bf S})$ for each
${\bf S} \in {\cal D}({\cal G}_{{\bf V}\setminus{\bf A}})$.
We will omit ${\bf A}$ and ${\cal G}$ from the subscript if they are obvious,
to yield ${\bf D}_{\bf S}$, and
let ${\bf A}_f = {\bf A} \setminus
\bigcup_{{\bf S} \in {\cal D}({\cal G}_{{\bf V} \setminus {\bf A}})}
{\bf D}_{{\bf S}}$.

In words, ${\cal D}({\cal G}_{{\bf V} \setminus {\bf A}})$ is the
districts in a graph where treatments ${\bf A}$ are removed.  For instance,
in Fig. \ref{fig:triangle2} (d), with treatment $A$, these districts are
$\{ M \}$ and $\{ Y \}$.  For each such district ${\bf S}$, ${\bf D}_{\bf S}$
is a (possibly larger) district containing all of ${\bf S}$ in a graph
containing ancestors of ${\bf S}$.  For instance, ${\bf D}_{\{ Y \}}$ in
Fig. \ref{fig:triangle2} (d) is $\{ A, Y \}$.  ${\bf A}_f$ is all treatments
not in any such ${\bf D}_{\bf S}$.  Since $A$ is the only treatment in
Fig. \ref{fig:triangle2} (d) and is in ${\bf D}_{\{ Y \}}$, ${\bf A}_f = \{\}$
in this case.

\begin{lem}
If
$\left(
\forall {\bf S}_1,{\bf S}_2 \in {\cal D}({\cal G}_{{\bf V}\setminus{\bf A}})
\right)
({\bf D}_{{\bf S}_1} \cap {\bf D}_{{\bf S}_2} \neq \emptyset) \Rightarrow
({\bf S}_1 = {\bf S}_2)$,
then the sets
$\{ {\bf D}_{\bf S} \mid
{\bf S} \in {\cal G}({\cal G}_{{\bf V} \setminus {\bf A}}) \}$
partition ${\bf V} \setminus {\bf A}_f$.
\label{lem:dis-part}
\end{lem}

Given Lemma \ref{lem:dis-part}, for every $V \in {\bf V} \setminus {\bf A}_f$,
let ${\bf D}_V = {\bf D}_{\bf S}$ for the unique ${\bf D}_{\bf S}$ such that
$V \in {\bf D}_{\bf S}$.

The following lemma gives two conditions necessary for {\bf ID} to give a
g-functional output.  First, any district ${\bf S} \in
{\cal G}_{{\bf V}\setminus{\bf A}}$ must not have parents not in ${\bf S}$
as elements of ${\bf D}_{\bf S}$, and second the sets ${\bf D}_{\bf S}$ must
partition ${\bf V} \setminus {\bf A}_f$ as in Lemma \ref{lem:dis-part}.
This is satisfied by Fig. \ref{fig:triangle2} (d), since
${\bf D}_{\{ Y \}} = \{ Y, A \}$, ${\bf D}_{\{ M \}} = \{ M \}$, and
$\pa_{\cal G}(\{M\}) = \{A\}$, $\pa_{\cal G}(\{Y\}) = \{ M \}$.

\begin{lem}
If the inputs ${\bf A},{\bf Y},{\cal G}$ to {\bf ID} are such that
\begin{itemize}
\item[1]
$\left(\forall {\bf S} \in {\cal D}({\cal G}_{{\bf V} \setminus {\bf A}})
\right)$,
$(\pa_{{\cal G}}({\bf S}) \setminus {\bf S}) \cap
{\bf D}_{\bf S} = \emptyset$,
and
\item[2] $\left(\forall {\bf S}_1,{\bf S}_2 \in
{\cal D}({\cal G}_{{\bf V}\setminus{\bf A}}) \right), ({\bf D}_{{\bf S}_1}
\cap {\bf D}_{{\bf S}_2} \neq \emptyset) \Rightarrow ({\bf S}_1 = {\bf S}_2)$,
\end{itemize}
then $p({\bf Y}({\bf a}) = {\bf y})$ is identified by a g-functional
\begin{align}
\sum_{{\bf v}_{{\bf V} \setminus ({\bf Y} \cup {\bf A}_f)}}
	\prod_{V \in {\bf V} \setminus {\bf A}_f}
		p(({\bf y} \cup {\bf v})_V \mid
		{\bf a}_{\pre_{\prec_{{\cal G}}}(V) \cap
		({\bf A} \setminus {\bf D}_V)},
		({\bf y} \cup {\bf v})
		_{\pre_{\prec_{{\cal G}}}(V) \setminus
		({\bf A} \setminus {\bf D}_V)}).
\label{eqn:g-functional}
\end{align}
\label{lem:g-func}
\end{lem}

Finally, given that preconditions given by Lemma \ref{lem:g-func} are satisfied
by an ADMG ${\cal G}$, the following result claims we can modify ${\cal G}$
into a DAG ${\cal G}^{\dag}$, where there is some edge intervention with a
response identified by the same g-functional as given by lemma \ref{lem:g-func}.
This DAG for Fig. \ref{fig:triangle2} (d) is Fig. \ref{fig:triangle2} (c).

\begin{lem}
For an ADMG ${\cal G}$ with vertex set ${\bf V}$, fix disjoint
${\bf Y},{\bf A} \subseteq {\bf V}$ that satisfy the preconditions of
lemma \ref{lem:g-func}.
Then there exists a DAG ${\cal G}^{\dag}$ with vertex set ${\bf V}$,
and an edge intervention
$\eta_{{\mathfrak a}_{\bm\alpha}}$ on a set of edges
${\bm\alpha}$ in ${\cal G}^{\dag}$ such that
$p({\bf Y}({\mathfrak a}_{\bm\alpha}))$ is identified under the MWM for
${\cal G}^{\dag}$ via a margin of the functional in (\ref{eqn:g-formula-edge})
that is equal to the identifying g-functional for
$p({\bf Y}({\bf a}))$ in terms of $p({\bf V})$ in ${\cal G}$.
\label{lem:g-func-as-edge-g}
\end{lem}

A natural question raised by lemma
\ref{lem:g-func-as-edge-g} is the converse -- is it the case that every
identifying functional for an edge intervention corresponds
to an identifying functional of a causal effect via {\bf ID}.
We leave this question for future work.

The fact that a class of causal effects identified via a g-functional,
even those effects with unobserved causes of treatments,
corresponds to responses to edge interventions in a DAG
gives an additional reason to study estimation theory of the edge g-formula
(\ref{eqn:g-formula-edge}).  Furthermore, this
connection gives another setting in which front-door type functionals may
arise -- the context of mediation analysis where the baseline treatment is
not a constant value, but a naturally occurring value in the population.

\section{A Multiply Robust Estimator for a Special Case of the Edge G-Formula}
\label{sec:estimator}

\begin{figure}
\begin{center}
\begin{tikzpicture}[>=stealth, node distance=0.9cm]
    \tikzstyle{format} = [draw, very thick, circle, minimum size=5.0mm,
	inner sep=0pt]

	\begin{scope}
		\path[->, very thick]
			node[format] (a) {$A$}
			node[format, right of=a] (m) {$M$}
			node[format, right of=m] (y) {$Y$}
			node[format, gray, above of=m,yshift=-0.1cm] (h) {$H$}

			(a) edge[blue] (m)
			(m) edge[blue] (y)

			(h) edge[blue] (a)			
			(h) edge[blue] (y)

			node[below of=m, yshift=0.3cm] (l) {$(a)$}
		;
	\end{scope}
	\begin{scope}[xshift=3.0cm]
		\path[->, very thick]
			node[format] (a) {$A$}
			node[format, right of=a] (m) {$M$}
			node[format, right of=m] (y) {$Y$}
			node[format, above of=m,yshift=-0.1cm] (h) {$C$}

			(a) edge[blue] (m)
			(m) edge[blue] (y)

			(h) edge[blue] (a)			
			(h) edge[blue] (y)
			(h) edge[blue] (m)

			(a) edge[blue, bend right] (y)

			node[below of=m, yshift=0.3cm] (l) {$(b)$}
		;
	\end{scope}
	\begin{scope}[xshift=6.0cm]
		\path[->, very thick]
			node[format] (a) {$A$}
			node[format, right of=a] (m) {$M$}
			node[format, right of=m] (y) {$Y$}

			(a) edge[blue, bend left] (y)
			(a) edge[blue] (m)
			(m) edge[blue] (y)

			node[below of=m, yshift=0.3cm] (l) {$(c)$}
		;
	\end{scope}
	\begin{scope}[xshift=9.0cm]
		\path[->, very thick]
			node[format] (a) {$A$}
			node[format, right of=a] (m) {$M$}
			node[format, right of=m] (y) {$Y$}

			(a) edge[<->, red, bend left] (y)
			(a) edge[blue] (m)
			(m) edge[blue] (y)

			node[below of=m, yshift=0.3cm] (l) {$(d)$}
		;
	\end{scope}

\end{tikzpicture}
\end{center}
\caption{
(a) A hidden variable
DAG where the causal effect $p(Y(a) = y)$ is identified via the front-door
formula $\sum_{m,a'} p(y \mid a',m) p(m \mid a) p(a')$.
(b) A DAG for a simple setting in mediation analysis where multiply robust
estimators for functionals derived from (\ref{eqn:g-formula-edge}) for
$Y((aY)_{\to},(a'M)_{\to})$ are known.
(c) A DAG where $Y((aM)_{\to})$ is identified via the front-door formula in
(a).
(d) A latent projection ADMG of the DAG in (a) onto $\{ A, M, Y \}$.
}
\label{fig:triangle2}
\end{figure}
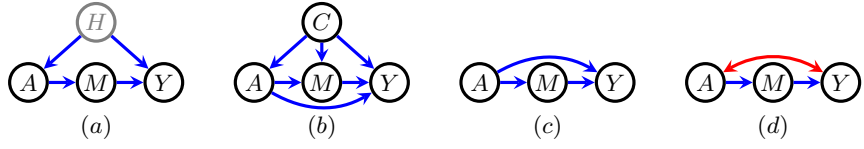

We have shown that the edge g-formula (\ref{eqn:g-formula-edge})
encodes a wide class of identified targets in causal inference.
Here we give an example of how a response to an edge-consistent path
intervention is represented as an edge intervention, identified via a marginal
of (\ref{eqn:g-formula-edge}), and re-expressed as a contrast parameter for
which an estimator exists which is robust to misspecification of parts of the
likelihood function.  We consider discrete state spaces, but
extensions to continuous state spaces are straightforward in this case.

Consider the graph in Fig. \ref{fig:triangle2} (b), which represents a simple
mediation setting, with $A$ an exposure, $Y$ an outcome, $M$ a mediator, and
$C$ a set of baseline covariates.  We might be interested in a direct or
indirect effect of $A$ on $Y$.  As discussed in section \ref{sec:targets},
we may represent such effects as contrasts obtained from a response to a path
intervention $p(Y((aMY)_{\to}, (a'Y)_{\to}))$.  This path intervention is
natural, and edge consistent, and the response of $Y$ to it is equal to
the response to an edge intervention $p(Y((aM)_{\to},(a'Y)_{\to}))$, which
is identified as a marginal of (\ref{eqn:g-formula-edge}), namely
$\sum_{c} p(Y \mid a', m, c) p(m \mid a, c) p(c)$.  Let
$\Upsilon(a,a',c) = \sum_m \mathbb{E}(Y \mid a', m, c) \cdot p(m \mid a, c)$.
Then the mean response is $\Phi(a,a') = \sum_{c} \Upsilon(a,a',c) \cdot p(c)$,
and the \emph{efficient influence function} of $\Phi(a,a')$
under the saturated model ${\cal P}_s$, that is the set of all densities
$p(Y,A,M,C)$, is
\begin{align*}
U^{\mathrm{eff}}_{{\cal P}_s}(\Phi(a,a')) =
& \frac{\mathbb{I}(A=a) p(M \mid a',C)}{p(a\mid C) p(M \mid a, C)}
\{ Y - \mathbb{E}(Y \mid C,M,a) \} +\\
& \frac{\mathbb{I}(A=a')}{p(a'\mid C)}\{ \mathbb{E}(Y\mid C,M,a) -
\Upsilon(a,a',C)\}
+ \Upsilon(a,a',C) - \Phi(a,a'),
\end{align*}
where $\mathbb{I}(.)$ is the indicator function for an event
\cite{tchetgen12semi2}.

To represent direct and indirect effects as contrasts, we also need to consider
the response of $Y$ to $A$ being set to $a$ for the purposes of all
pathways from $A$ to $Y$, which simply corresponds to $p(Y(a))$, which is
identified via a marginal of (\ref{eqn:g-formula}), namely
$\sum_{c} p(Y \mid a,c) p(c)$.  The mean response is then
$\Phi(a,a) = \sum_{c} \mathbb{E}(Y \mid a,c) p(c)$.
The efficient influence function of $\Phi(a,a)$ under the
saturated model ${\cal P}_s$ is simply
 $U^{\mathrm{eff}}_{{\cal P}_s}(\Phi(a,a))$, which simplifies to
\[
\frac{\mathbb{I}(a)}{p(a\mid C)}\{ Y - \Upsilon(a,a,C) \} + \Upsilon(a,a,C) -
\Phi(a,a),
\]
the efficient influence function derived in the context of total effects in
\cite{robins94estimation}.

Natural direct and indirect effects may be defined on the difference scale as
$\Phi(a,a) - \Phi(a,a')$, and $\Phi(a,a') - \Phi(a',a')$.  Alternatively,
for binary outcomes we may also define such effects in a natural way on the
risk ratio or odds ratio scale.

Estimating these parameters using an unrestricted likelihood is not a feasible
strategy in settings with a high dimensional vector of baseline covariates,
which means we must resort to modeling.
An approach in \cite{tchetgen12semi2} is to
assume models $\{ \mathbb{E}^{\mathrm{par}}(Y \mid a,m,c ; \hat{\alpha}),
f^{\mathrm{par}}(m \mid a,c ; \hat{\beta}),
f^{\mathrm{par}}(a\mid c ; \hat{\gamma}) \}$, and use a substitution estimator
which solves the estimating equations
\[
\mathbb{P}_n\left( \hat{U}^{\mathrm{eff}}_{{\cal P}_s}(\Phi(a,a')) \right)
= 0,
\]
where $\mathbb{P}_n(.)$ is the empirical average (for sample size $n$),
and $\hat{U}^{\mathrm{eff}}_{{\cal P}_s}$ is equal to
${U}^{\mathrm{eff}}_{{\cal P}_s}$ evaluated at
$\{ \mathbb{E}^{\mathrm{par}}(Y \mid a,m,c ; \hat{\alpha}),
f^{\mathrm{par}}(m \mid a,c ; \hat{\beta}),
f^{\mathrm{par}}(a\mid c ; \hat{\gamma}) \}$.

The resulting estimator exhibits the property of \emph{triple robustness},
that is it remains consistent in the union model where any two of the
above three parametric models is correct.  This estimator is combined with a
similarly defined doubly
robust estimator for $\Phi(a,a)$ derived in \cite{robins94estimation} to
yield a triply robust estimator for the direct and indirect parameters on
the difference scale.  This was extended to the semi-parametric models for
direct effects on the additive and multiplicative scales
\cite{tchetgen14semi}.

Since our results show that the edge g-formula encompasses a wide range of
causal inference targets, including effects of treatments on the multiply
treated, path-specific effects, and causal effects with unobserved causes
of treatments, an interesting avenue of future work is to generalize
estimation theory for simple instances of the edge g-formula, like above, to
more general cases, for instance longitudinal cases like that shown in Fig.
\ref{fig:triangle} (b).

\section{Discussion}
\label{sec:discussion}

We have defined an inclusion hierarchy of interventions associated with
graphical features:
node interventions corresponding to standard treatment interventions, edge
interventions corresponding to intervening on a portion of the treatment
mechanism associated with a particular outgoing edge, and path interventions
corresponding to intervening on a portion of the treatment mechanism associated
with a particular outgoing causal
pathway.  We have shown that a variety of causal inference targets of interest,
including effects of treatment on the multiply treated, and path-specific
effects can be viewed as special cases of responses to path interventions.
In addition, we have shown that edge interventions are in some sense naturally
associated with the MWM of Pearl as the responses to such interventions are
naturally identified under the assumptions of this model, just as node
interventions are naturally associated with the SWM of Robins.
The question of whether a particular causal inference target is identified,
and under what model thus reduces to expressing the target as a path
intervention, and then considering whether the path intervention is natural,
and whether it can be re-expressed as an edge intervention or a node
intervention.  This process is summarized in a flowchart shown in Fig.
\ref{fig:flowchart}.

\begin{figure}
\begin{center}
\begin{tikzpicture}[>=stealth, node distance=1.3cm]
    \tikzstyle{format} = [draw, thick, ellipse, minimum size=5.0mm,
	inner sep=0pt, text width=2cm]
    \tikzstyle{decision} = [draw, thick, diamond, minimum size=5.0mm,
	inner sep=0pt, text width=1.4cm]

	\begin{scope}[xshift=8.0cm]
		\path[->]
			node[decision] (natural)
			{\hspace{1.0mm}natural?}

			node[decision, text width=1.5cm,
				above right of=natural, xshift=2.0cm]
			(edge-cons)
			{\hspace{4.0mm}edge\\consistent?}

			node[format, below right of=natural, xshift=2.0cm]
				(not-id1)
				{\hspace{1.5mm}not id under\\
				\hspace{5.4mm}MWM}

			node[format, right of=edge-cons, xshift=2.0cm]
				(edge-int)
			{\hspace{4.0mm}id under\\
				\hspace{5.4mm}MWM}

			node[decision, text width=1.5cm,
				right of=edge-int, xshift=2.0cm]
			(node-cons)
			{\hspace{4.0mm}node\\consistent?}

			node[format, right of=not-id1, xshift=2.0cm]
				(not-id2)
				{\hspace{1.5mm}not id under\\
				\hspace{5.4mm}SWM}

			node[format, right of=not-id2, xshift=2.0cm]
				(node-int)
				{\hspace{4.0mm}id under\\
				\hspace{5.4mm}SWM}

			(natural) edge[thick] node[above] {no} (not-id1)
			(natural) edge[thick] node[below] {yes} (edge-cons)
			(edge-cons) edge[thick] node[right] {no} (not-id1)
			(edge-cons) edge[thick] node[above] {yes}
				(edge-int)

			(edge-int) edge[thick]
			(node-cons)

			(node-cons) edge[thick] node[right] {no}
				(not-id2)

			(node-cons) edge[thick] node[right] {yes}
				(node-int)
		;
	\end{scope}

\end{tikzpicture}
\end{center}
\caption{A flowchart for identification results for path interventions under
the MWM and the SWM.}
\label{fig:flowchart}
\end{figure}
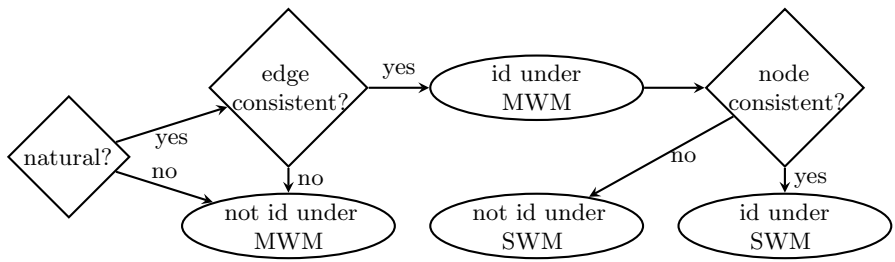


An obvious extension of our work is to consider identification of responses
in our hierarchy in hidden variable DAG models in terms of observed marginal
distributions.  Existing results on mediation analysis \cite{shpitser13cogsci}
and ETT identification \cite{shpitser09ett} would be subsumed as special cases
under this framework, but it would entail novel identification results for
any new target expressible as a path intervention response.
In addition, an interesting question is whether
\emph{all} identifying functionals of Tian's algorithm {\bf ID} correspond to
some sort of identified response to an edge intervention, although possibly
not in a DAG but an ADMG.  If true, this would recast \emph{any} identified
causal effect as a certain type of identified mediated effect.

While estimation theory of functionals derived from the
extended g-formula (\ref{eqn:g-formula}) has received attention in
the literature \cite{robins04effects},
estimators for functionals obtained
from the edge g-formula (\ref{eqn:g-formula-edge}) are known only
in very special cases such as the point treatment setting we discussed
in section \ref{sec:estimator} \cite{tchetgen12semi2}.
As we have shown in this paper, developing estimators for general functionals
obtained from the edge g-formula (\ref{eqn:g-formula-edge}) results in
estimators for a wide class of targets of interest in causal inference,
including path-specific effects, effects of treatment on the multiply treated,
effects of treatments on the indirectly treated,
and causal effects in the presence of unobserved causes of treatments.

Our results thus not only provide a unifying view of identification, under
various models, of a large class of targets of interest in causal inference,
but also motivate the development of estimation theory for a more general
functional than the g-formula.

\begin{supplement}[id=suppA]
  \stitle{Supplementary Materials For: Causal Inference with a Graphical
	Hierarchy of Interventions}
  \slink[doi]{COMPLETED BY THE TYPESETTER}
  \sdatatype{.pdf}
  \sdescription{Our supplementary materials contain detailed arguments for
	most of our claims, and some auxiliary definitions, including the
	definition of the {\bf ID} algorithm.
	In addition, we provide a detailed
	rationale for the use of path interventions.}
\end{supplement}

\bibliographystyle{plainnat}
\bibliography{references}

\begin{thebibliography}{36}
\providecommand{\natexlab}[1]{#1}
\providecommand{\url}[1]{\texttt{#1}}
\expandafter\ifx\csname urlstyle\endcsname\relax
  \providecommand{\doi}[1]{doi: #1}\else
  \providecommand{\doi}{doi: \begingroup \urlstyle{rm}\Url}\fi

\bibitem[Avin et~al.(2005)Avin, Shpitser, and Pearl]{chen05ijcai}
Chen Avin, Ilya Shpitser, and Judea Pearl.
\newblock Identifiability of path-specific effects.
\newblock In \emph{International Joint Conference on Artificial Intelligence},
  volume~19, pages 357--363, 2005.

\bibitem[Balke and Pearl(1994)]{balke94countereval}
Alexander Balke and Judea Pearl.
\newblock Probabilistic evaluation of counterfactual queries.
\newblock In \emph{Proceedings of AAAI-94}, pages 230--237, 1994.

\bibitem[Baron and Kenny(1986)]{baron86mm}
Reuben~M. Baron and David~A. Kenny.
\newblock The moderator-mediator variable distinction in social psychology
  research: Conceptual, strategic, and statistical considerations.
\newblock \emph{Journal of Personality and Social Psychology}, 51:\penalty0
  1173--1182, 1986.

\bibitem[Hubbard and Laan(2008)]{population08hubbard}
Alan~E. Hubbard and Mark J. Van~Der Laan.
\newblock Population intervention models in causal inference.
\newblock \emph{Biometrika}, 95\penalty0 (1):\penalty0 35--47, 2008.

\bibitem[Imai et~al.(2013)Imai, Tingley, and Yamamoto]{imai13experimental}
Kosuke Imai, Dustin Tingley, and Teppei Yamamoto.
\newblock Experimental designs for identifying causal mechanisms.
\newblock \emph{Journal of the Royal Statistical Society, series (A)}, 176(1),
  2013.

\bibitem[Moodie et~al.(2007)Moodie, Richardson, and
  Stephens]{moodie07demystifying}
Erica E.~M. Moodie, Thomas~S. Richardson, and David~A. Stephens.
\newblock Demystifying optimal dynamic treatment regimes.
\newblock \emph{Biometrics}, 63\penalty0 (2):\penalty0 447--455, 2007.

\bibitem[Murphy(2003)]{murphy03optimal}
Susan~A. Murphy.
\newblock Optimal dynamic treatment regimes.
\newblock \emph{Journal of the Royal Statistical Society (Series B)},
  65:\penalty0 331--366, 2003.

\bibitem[Neyman(1923)]{neyman23app}
Jerzy Neyman.
\newblock Sur les applications de la thar des probabilities aux experiences
  agaricales: Essay des principle. excerpts reprinted (1990) in {E}nglish.
\newblock \emph{Statistical Science}, 5:\penalty0 463--472, 1923.

\bibitem[Pearl(1988)]{pearl88probabilistic}
Judea Pearl.
\newblock \emph{Probabilistic Reasoning in Intelligent Systems}.
\newblock Morgan and Kaufmann, San Mateo, 1988.

\bibitem[Pearl(2001)]{pearl01direct}
Judea Pearl.
\newblock Direct and indirect effects.
\newblock In \emph{Proceedings of UAI-01}, pages 411--420, 2001.

\bibitem[Pearl(2009)]{pearl09causality}
Judea Pearl.
\newblock \emph{Causality: Models, Reasoning, and Inference}.
\newblock Cambridge University Press, 2 edition, 2009.
\newblock ISBN 978-0521895606.

\bibitem[Pearl(2011)]{pearl11cmf}
Judea Pearl.
\newblock The causal mediation formula -- a guide to the assessment of pathways
  and mechanisms.
\newblock Technical Report R-379, Cognitive Systems Laboratory, University of
  California, Los Angeles, 2011.

\bibitem[Richardson and Robins(2013)]{thomas13swig}
Thomas~S. Richardson and Jamie~M. Robins.
\newblock Single world intervention graphs ({SWIG}s): A unification of the
  counterfactual and graphical approaches to causality.
\newblock \emph{preprint:
  \url{http://www.csss.washington.edu/Papers/wp128.pdf}}, 2013.

\bibitem[Robins(1986)]{robins86new}
James~M. Robins.
\newblock A new approach to causal inference in mortality studies with
  sustained exposure periods -- application to control of the healthy worker
  survivor effect.
\newblock \emph{Mathematical Modeling}, 7:\penalty0 1393--1512, 1986.

\bibitem[Robins(1987)]{robins87graphical}
James~M. Robins.
\newblock A graphical approach to the identification and estimation of causal
  parameters in mortality studies with sustained exposure periods.
\newblock \emph{Journal of Chronic Diseases}, 40:\penalty0 139--161, 1987.

\bibitem[Robins and Greenland(1992)]{robins92effects}
James~M. Robins and Sander Greenland.
\newblock Identifiability and exchangeability of direct and indirect effects.
\newblock \emph{Epidemiology}, 3:\penalty0 143--155, 1992.

\bibitem[Robins and Richardson(2010)]{robins10alternative}
James~M. Robins and Thomas~S. Richardson.
\newblock Alternative graphical causal models and the identification of direct
  effects.
\newblock \emph{Causality and Psychopathology: Finding the Determinants of
  Disorders and their Cures}, 2010.

\bibitem[Robins et~al.(1994)Robins, Rotnitzky, and Zhao]{robins94estimation}
James~M. Robins, Andrea Rotnitzky, and Lue~P. Zhao.
\newblock Estimation of regression coefficients when some regressors are not
  always observed.
\newblock \emph{Journal of the American Statistical Association}, 89:\penalty0
  846--866, 1994.

\bibitem[Robins et~al.(2000)Robins, Hernán, and Brumback]{robins00marginal}
James~M. Robins, Miguel Hernán, and Babette Brumback.
\newblock Marginal structural models and causal inference in epidemiology.
\newblock \emph{Epidemiology}, 11(5):\penalty0 550--560, 2000.

\bibitem[Robins et~al.(2004)Robins, Hern\'an, and Siebert]{robins04effects}
James~M. Robins, Miguel~A. Hern\'an, and Uwe Siebert.
\newblock Effects of multiple interventions.
\newblock \emph{Comparative quantification of health risks : global and
  regional burden of disease attributable to selected major risk factors},
  2\penalty0 (28):\penalty0 2191--2230, 2004.

\bibitem[Rubin(1974)]{rubin74potential}
Donald~B. Rubin.
\newblock Estimating causal effects of treatments in randomized and
  non-randomized studies.
\newblock \emph{Journal of Educational Psychology}, 66:\penalty0 688--701,
  1974.

\bibitem[Rubin(2004)]{rubin04direct}
Donald~B. Rubin.
\newblock Direct and indirect causal effects via potential outcomes.
\newblock \emph{Scandinavian Journal of Statistics}, 31:\penalty0 161--170,
  2004.

\bibitem[Rubin(2005)]{rubin05causal}
Donald~B. Rubin.
\newblock Causal inference using potential outcomes: design, modeling,
  decisions.
\newblock \emph{Journal of the American Statistical Association}, 100:\penalty0
  322--331, 2005.

\bibitem[Shpitser(2013)]{shpitser13cogsci}
Ilya Shpitser.
\newblock Counterfactual graphical models for longitudinal mediation analysis
  with unobserved confounding.
\newblock \emph{Cognitive Science ({R}umelhart special issue)}, 37:\penalty0
  1011--1035, 2013.

\bibitem[Shpitser and Pearl(2006{\natexlab{a}})]{shpitser06id}
Ilya Shpitser and Judea Pearl.
\newblock Identification of joint interventional distributions in recursive
  semi-{M}arkovian causal models.
\newblock In \emph{National Conference on Artificial Intelligence}, volume~21.
  AUAI Press, 2006{\natexlab{a}}.

\bibitem[Shpitser and Pearl(2006{\natexlab{b}})]{shpitser06idc}
Ilya Shpitser and Judea Pearl.
\newblock Identification of conditional interventional distributions.
\newblock In \emph{Uncertainty in Artificial Intelligence}, volume~22. AUAI
  Press, 2006{\natexlab{b}}.

\bibitem[Shpitser and Pearl(2008)]{shpitser07hierarchy}
Ilya Shpitser and Judea Pearl.
\newblock Complete identification methods for the causal hierarchy.
\newblock \emph{Journal of Machine Learning Research}, 9(Sep):\penalty0
  1941--1979, 2008.

\bibitem[Shpitser and Pearl(2009)]{shpitser09ett}
Ilya Shpitser and Judea Pearl.
\newblock Effects of treatment on the treated: identification and
  generalization.
\newblock In \emph{Uncertainty in Artificial Intelligence}, volume~25. AUAI
  Press, 2009.

\bibitem[Shpitser and Tchetgen(2015)]{shpitser15nep}
Ilya Shpitser and Eric~Tchetgen Tchetgen.
\newblock Supplementary materials for: Causal inference with a graphical
  hierarchy of interventions, 2015.

\bibitem[Spirtes et~al.(2001)Spirtes, Glymour, and
  Scheines]{spirtes01causation}
Peter Spirtes, Clark Glymour, and Richard Scheines.
\newblock \emph{Causation, Prediction, and Search}.
\newblock Springer Verlag, New York, 2 edition, 2001.
\newblock ISBN 978-0262194402.

\bibitem[Tchetgen and Shpitser(2012)]{tchetgen12semi2}
Eric J.~Tchetgen Tchetgen and Ilya Shpitser.
\newblock Semiparametric theory for causal mediation analysis: efficiency
  bounds, multiple robustness, and sensitivity analysis.
\newblock \emph{Annals of Statistics (in Press)}, 2012.

\bibitem[Tchetgen and Shpitser(2014)]{tchetgen14semi}
Eric J.~Tchetgen Tchetgen and Ilya Shpitser.
\newblock Estimation of a semiparametric natural direct effect model
  incorporating baseline covariates.
\newblock \emph{Biometrika (in Press)}, 2014.

\bibitem[Tian and Pearl(2002)]{tian02on}
Jin Tian and Judea Pearl.
\newblock On the testable implications of causal models with hidden variables.
\newblock In \emph{Uncertainty in Artificial Intelligence}, volume~18, pages
  519--527. AUAI Press, 2002.

\bibitem[Verma and Pearl(1990)]{verma90equiv}
Thomas~S. Verma and Judea Pearl.
\newblock Equivalence and synthesis of causal models.
\newblock Technical Report R-150, Department of Computer Science, University of
  California, Los Angeles, 1990.

\bibitem[Wright(1921)]{wright21correlation}
Sewall Wright.
\newblock Correlation and causation.
\newblock \emph{Journal of Agricultural Research}, 20:\penalty0 557--585, 1921.

\bibitem[Young et~al.(2014)Young, Hernan, and Robins]{identification14jessica}
Jessica~G. Young, Miguel~A. Hernan, and James~M. Robins.
\newblock Identification, estimation and approximation of risk under
  interventions that depend on the natural value of treatment using
  observational data.
\newblock \emph{Epidemiologic Methods}, 3\penalty0 (1):\penalty0 1--19, 2014.

\end{thebibliography}


\end{document}